\documentclass[sigconf]{acmart}

\AtBeginDocument{%
  \providecommand\BibTeX{{%
    \normalfont B\kern-0.5em{\scshape i\kern-0.25em b}\kern-0.8em\TeX}}}

\setcopyright{rightsretained}
\copyrightyear{2019}
\acmYear{2019}
\acmDOI{}

\acmPrice{}
\acmISBN{}

\usepackage{array}
\usepackage{mathrsfs}
\usepackage{float}

\usepackage{euscript}
\usepackage{amsmath}
\usepackage{graphicx}
\usepackage{epstopdf}
\usepackage{setspace}
\usepackage{mystyle}
\usepackage{algorithm}
\usepackage[noend]{algorithmic}

\usepackage{multirow}
\usepackage{rotating}

\usepackage{diagbox}

\newcommand{\R}{\ensuremath{\mathbb{R}}}
\newcommand{\dir}[1]{\ensuremath{\mathsf{d}#1}}
\newcommand{\dS}{\ensuremath{\mathtt{d}_{\mathsf{dE}}}}
\newcommand{\dQ}{\ensuremath{\mathtt{d}_{Q}}}
\newcommand{\dQi}{\ensuremath{\mathtt{d}_{Q_1}}}
\newcommand{\dQii}{\ensuremath{\mathtt{d}_{Q_2}}}
\newcommand{\dgen}{\ensuremath{\mathtt{d}}}
\newcommand{\ddQ}{\ensuremath{\bar{\mathtt{d}}_{Q}}}
\newcommand{\DQ}{\ensuremath{\tilde{\mathtt{d}}_{Q}}}
\newcommand{\dQk}{\ensuremath{\mathtt{d}^\leftrightarrow_{Q}}}

\newcommand{\dQW}{\ensuremath{\mathtt{d}_{Q,W}}}

\newcommand{\dQki}{\ensuremath{\mathtt{d}^\leftrightarrow_{Q_1}}}
\newcommand{\dQkii}{\ensuremath{\mathtt{d}^\leftrightarrow_{Q_2}}}
\newcommand{\dQWk}{\ensuremath{\mathtt{d}^\leftrightarrow_{Q,W}}}

\renewcommand{\c}[1]{\ensuremath{\EuScript{#1}}}
\newcommand{\PreserveBackslash}[1]{\let\temp=\\#1\let\\=\temp}
\newcolumntype{C}[1]{>{\PreserveBackslash\centering}p{#1}}
\newcolumntype{R}[1]{>{\PreserveBackslash\raggedleft}p{#1}}
\newcolumntype{L}[1]{>{\PreserveBackslash\raggedright}p{#1}}

\begin{document}

\title{Simple Distances for Trajectories via Landmarks}

\author{Jeff M. Phillips
}
\email{jeffp@cs.utah.edu }
\affiliation{%
  \institution{School of Computing, University of Utah.}
  \city{Salt Lake City, Utah}
  \country{USA}
}

\author{Pingfan Tang
}
\email{tang1984@cs.utah.edu}
\affiliation{%
  \institution{School of Computing, University of Utah.}
  \city{Salt Lake City, Utah}
  \country{USA}
}

\begin{abstract}
We develop a new class of distances for objects including lines, hyperplanes, and trajectories, based on the distance to a set of landmarks.  These distances easily and interpretably map objects to a Euclidean space, are simple to compute, and perform well in data analysis tasks.  For trajectories, they match and in some cases significantly out-perform all state-of-the-art other metrics, can effortlessly be used in $k$-means clustering, and directly plugged into approximate nearest neighbor approaches which immediately out-perform the best recent advances in trajectory similarity search by several orders of magnitude.
These distances do not require a geometry distorting dual (common in the line or halfspace case) or complicated alignment (common in trajectory case).  We show reasonable and often simple conditions under which these distances are metrics.
\end{abstract}




\maketitle

\section{Introduction}
The \emph{choice} of a distance is often the most important modeling decision in any data analysis task.  This choice is what determines which objects are close and which are far.  However, this task is often taken lightly or made just based on what provides the simplest or easiest to compute option.

In this paper we explore what we believe to be a new and natural family of distances between objects, focusing on two cases when the objects are hyperplanes (e.g., regressors or separators), or when they are trajectories.  Our proposed distance $\dQ$ uses a set $Q$ of landmark points, which could be the dataset that regressors or separators are trained on, or in the case of trajectories these may be points of interest for which a trajectory passing nearby has specific meaning.  However, in a general case, $Q$ can be chosen as arbitrary or random points placed to cover a domain of focus.
Then the new distances, instead of being directly between the objects themselves, are based on how they interact with the set of landmarks.  In the simplest variant, for $n$ landmarks $Q$, for any object $J$ we create an $n$-dimensional vector  $v_J = (v_1, v_2, \ldots, v_n)$ of  the distance from $q_i \in Q$ to $J$, and the distance between two objects $J_1$ and $J_2$ is the Euclidean distance between the vectors $\|v_{J_1} - v_{J_2}\|$.
In other words, we \emph{vectorize} the distance between complex objects.

In this paper we explore several variants of this formulation, derive convenient mathematical properties, and demonstrate its efficacy in several data analysis scenarios.

\paragraph*{Key properties of a distance.}
A definition of a distance $\dgen$ is the key building block in most data analysis tasks.  For instance, it is at the heart of any assignment-based clustering (e.g., $k$-means) or for nearest-neighbor searching and analysis.  We can also define a radial-basis kernel $K(p,q) = \exp(-\dgen(p,q)^2)$ (or similarly), which is required for kernel SVM classification, kernel regression, and kernel density estimation.  A change in the distance, directly affects the meaning and modeling inherent in each of these tasks.  So the first consideration in choosing a distance should always be, does it capture the properties between the objects that matter?

As we will observe, by having a distance depend on a set of landmarks $Q$, then we can tune it to focus on certain regions.  In the case of regressors or separators (e.g., infinite lines, hyperplanes) this makes sure the distance is determined by how these infinite objects interact with the support of the data.  In the case of trajectories, the distance can be adjusted to focus on one or more locations of interest (e.g., a sporting event or school) or regions of interest (e.g., how someone passes through an airport, but not how they get there), as opposed to its full geometry.

A generic desired property of a distance is that it should be a metric: for instance this is essential in the analysis for the Gonzalez algorithm~\cite{Gon85} for $k$-center clustering, and many other contexts such as nearest-neighbor searching.

Another generic goal is analyzing the distance's metric balls.  That is, given a set of objects  $\c{J}$ and a distance $\dgen : \c{J} \times \c{J} \to \R$, let $B(J,r) = \{J' \in \c{J} \mid \dgen(J,J') \leq r\}$ be a metric ball around $J \in \c{J}$ of radius $r$.  Then we can define a range space $(\c{J}, \c{R})$ where $\c{R} = \{B(J,r) \mid J \in \c{J}, r \geq 0\}$, and consider its VC-dimension~\cite{VC71}.  When the VC-dimension $\nu$ is small, it implies that the metric balls cannot interact with each other in a too complex way, indicating the distance is roughly as well-behaved as a $\nu$-dimensional Euclidean ball.  More directly, this implies, decision boundaries to classify objects can be learned with only $\eps$-fraction generalization error using $O(\nu/\eps \cdot \log(1/\eps))$ samples if the data is separable, or $O(\nu/\eps^2)$ samples if the data is not separable~\cite{LLS01}.  Similar bounds can be shown for other tasks such as preserving kernel density estimates derived from such distances~\cite{JoshiKommarajuPhillips2011}.  In other words, this ensures that many tasks are stable with respect to the underlying family of objects $\c{J}$.


\paragraph{Main results.}
We define a new data dependent distance $\dQ$ for trajectories and for linear models (e.g., regressors, separators) built from a landmark data set $Q$.  For the simpler cases of linear models (in Section \ref{distance 1D}), we show it is a metric as long as $Q$ is full rank.  We also show that its metric balls have VC-dimension bounded only by the ambient dimension and not on the size of $Q$.  We find this surprising because the distance corresponds to an embedding in $|Q|$-dimensional Euclidean space where an immediate bound for the VC-dimension is $|Q|+1$; and indeed this will be the best bound we have for most of the trajectory variants.
We show how to directly extend all of these definitions of lines to trajectories, with a somewhat unintuitive and restrictive distance measure $\dQ^\leftrightarrow$.

For the pressing scenario of trajectories, in Section \ref{sec:curves}, we introduce two more intuitive variants $\dQ$ and $\dQ^\pi$.  We describe simple conditions for $Q$ under which they are metrics.  We can immediately see that both distances are pseudometrics (they satisfy triangle inequality, and are symmetric, but might have distinct objects with distance $0$).  We show they satisfy the final $0$-property of a metric as long as the waypoints are distinct and $Q$ is sufficiently dense.  For all new variants we demonstrate that they are at the least as effective for classification tasks (via KNN classifiers) as compared to the best of $9$ other common metrics, and \emph{in some cases  significantly outperforms all of these measures}.  Moreover, the previous competing variants are typically significantly more complicated or computationally intensive, and may require parameter tuning.

In contrast to most of these trajectory distance alternatives, all of our proposed distances are very simple to compute and work with.  They map curves (or hyperplanes) to a $|Q|$-dimensional parameter space where Euclidean distance (or similar) is used.
In $\dQ$ for curves, each coordinate $v_i$ is the distance to the closest point on the curve from $q_i \in Q$.
In $\dQ^\pi$ each ``coordinate'' is actually the $d$ coordinates of the closest point on the curve (not just the distance).
In $\dQ^\leftrightarrow$ each ``coordinate'' $v_i$ is actually $k$ values, to the distance to the closest point on the $k$ lines extending the $k$ lines segments of the curve.
These mappings are effective with only $10$ or $20$ landmark points $Q$.  And because they have a familiar Euclidean structure, we can immediately invoke favorite algorithms in this space, from Lloyds for $k$-means clustering, linear and kernel SVM, and highly-engineered approximate nearest neighbor libraries.
In comparison to recent trajectory similarity search systems~\cite{XLP2017,SLB2018}, we show using $\dQ$ is much simpler and several orders of magnitude faster.

In summary, this paper introduces a family of metrics for regressors, separators, and piecewise-linear curves which are incredibly simple to use, provide a sketch vector in Euclidean space, have many other desirable mathematical properties, and perform as well as and often significantly better than any existing measure.

\section{Distance Between Lines and Hyperplanes} \label{distance 1D}

As a warm up to the general case, we define a new landmark-based distance $\dQ$ between two lines, and give the condition under which it is a metric.
Then we generalize to hyperplanes, and provide the general metric proof, the VC-dimension of metric ball proof, and some algorithmic implications.  We conclude with a direct extension to trajectories.

\subsection{Warm Up: Distance Between Lines}
We begin by reviewing alternatives, starting with the default \emph{dual Euclidean distance}.
Consider the least square regression problem in $\R^2$: given $Q=\{(x_1,y_1),\cdots,$ $(x_n,y_n)\}$ $\subset \R^2$, return a line $\ell:y=ax+b$ such that $(a,b)=\text{arg }\min_{(a,b)\in \R^2}$ $\sum_{i=1}^n(ax_i+b-y_i)^2$.
If $\ell_1: y=a_1x+b_1$ is an alternate fit to this data, then to measure the difference in these variants, we can define a distance between $\ell$ and $\ell_1$. A simple and commonly used distance (which we called the \emph{dual-Euclidean distance}) is
\[
\dS(\ell,\ell_1):=\sqrt{(a-a_1)^2+(b-b_1)^2}.
\]
This can be viewed as dualizing the lines into a space defined by their parameters (slope $a$ and intercept $b$), and then taking the Euclidean distance between these parametric points.
However, as shown in Figure \ref{fig:ddE}(Left), if both $\ell_1$ and $\ell_2$ have the same slope $a_1 = a_2$, and are offset the same amount from $\ell$ ($|b - b_1| = |b-b_2|$), then $\dS(\ell,\ell_1) = \dS(\ell, \ell_2)$, although intuitively $\ell_1$ does a much more similar job to $\ell$ with respect to $Q$ than does $\ell_2$.

\begin{figure}[b]
	\vspace{-.1in}
	\includegraphics[width=0.49\linewidth]{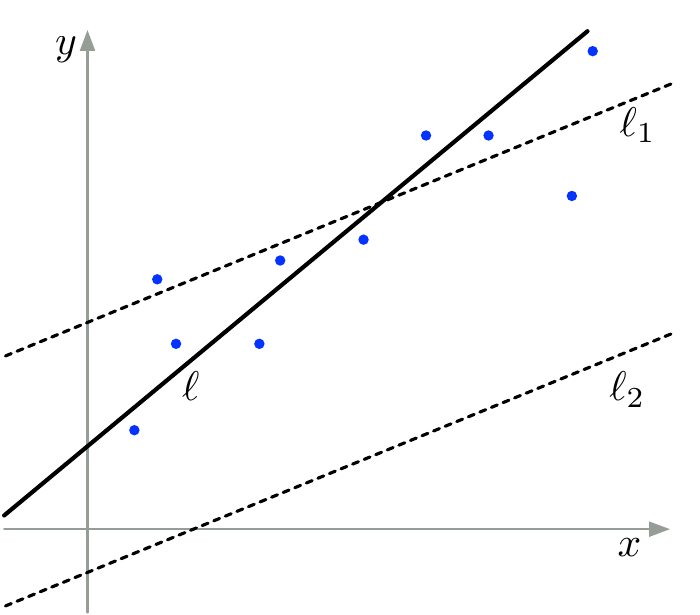}
     \includegraphics[width=0.49\linewidth]{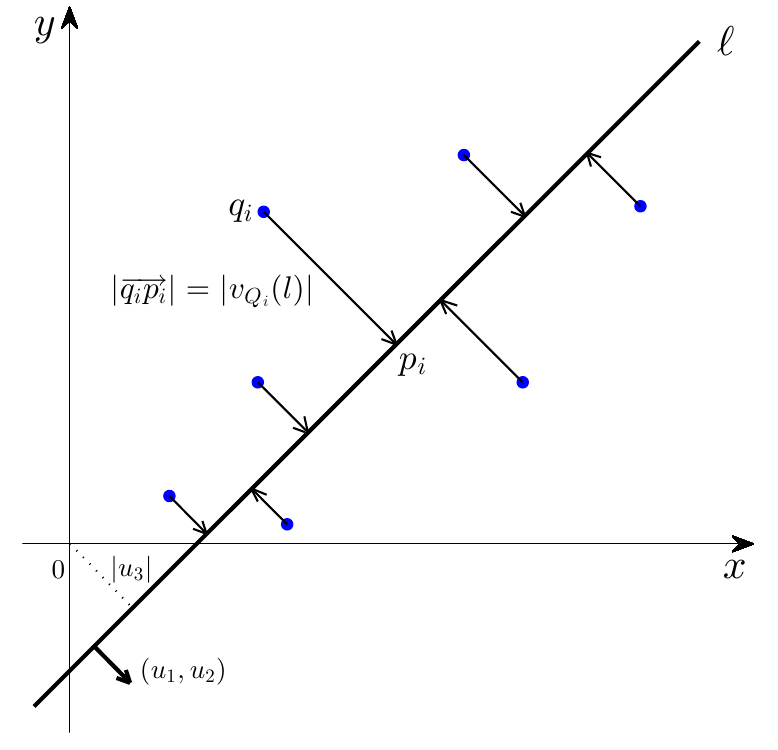}\\
	\vspace{-0.1in}
	\caption{
		Left: $\dS(\ell,\ell_1) = \dS(\ell, \ell_2)$, but which of $\ell_1$ and $\ell_2$ is more similar to $\ell$ with respect to $Q$?
	 	Right: Each $p_i$ is the projection of $q_i$ on $\ell$.}
	\label{fig:ddE}
\end{figure}

More generically, a geometric object is usually described by an (often compact) set in $\R^d$.
There are many ways to define and compute distances between such objects~\cite{AG96,AMWW88,EGDWS1988,EGG1990}.  These can be based on the minimum~\cite{EGDWS1988,EGG1990} or maximum (e.g., Hausdorff)~\cite{AG96,AMWW88} distance between objects.  We review more later in the context of trajectories in Section \ref{sec:related}.  For lines or hyperplanes which extend infinitely and may intersect at single points, such measures are not meaningful.

\paragraph*{Our formulation.}
Suppose $Q=\{q_1, q_2,
\cdots,  q_n\} \subset \R^2$ where $q_i$ has coordinates $(x_i,y_i)$ for $1\leq i \leq n$, and $\ell$ is a line in $\R^2$, then $\ell$ can be uniquely expressed as
\[
\ell =\{(x,y)\in \R^2 \ |\ \ u_1x+u_2y+u_3=0\},
\]
where $(u_1,u_2, u_3)\in \mathbb{U}^3$.
Here
$\b{U}^3 =\{u=(u_1,u_2,u_3)\in \R^3 \mid u_1^2+u_2^2=1$ and the first nonzero entry of $ u$  is positive$\}$,
is a canonical way to normalize $u$ where $(u_1, u_2)$ is unit normal vector and  $u_3$ is an offset parameter.
Let $v_{Q_i}(\ell)=u_1x_i+u_2y_i+u_3$; it is the signed distance from $q_i = (x_i, y_i)$ to the closest point on $\ell$.
Then $v_Q(\ell)=\left(v_{Q_1}(\ell),v_{Q_2}(\ell),\ldots,v_{Q_n}(\ell)\right)$ is the $n$-dimensional vector of these distances.
For two lines $\ell_1$, $\ell_2$ in $\R^2$, we can now define
\begin{align*}
\dQ(\ell_1,\ell_2)
&=
\Big\|\frac{1}{\sqrt{n}}(v_Q(\ell_1)-v_Q(\ell_2))\Big\|\\
&=
\Big(\sum_{i=1}^n\frac{1}{n}(v_{Q_i}(\ell_1)-v_{Q_i}(\ell_2))^2\Big)^{\frac{1}{2}}.
\end{align*}

As shown in Figure \ref{fig:ddE}(Right),
$|v_{Q_i}(\ell)|$ is the distance from $q_i$ to $\ell$. With the help of $Q$, we convert each line $\ell$ in $\R^2$ to point $\frac{1}{\sqrt{n}}v_Q(\ell)$ in $\R^n$, and use the Euclidean distance between two points to define the distance between the original two lines.  Via this Euclidean embedding, it directly follows that $\dQ$ is symmetric and follows the triangle inequality.
The following theorem shows, under reasonable assumptions of $Q$, no two different lines can be mapped to the same point in $\R^n$, so $\dQ$ is a metric.

\begin{theorem}\label{well-difined metric}
	Suppose in $Q=\{(x_1,y_1),(x_2,y_2),\cdots,$ $(x_n,y_n)\}\subset \R^2$ there are three non-collinear points, and $\c{L}=\{\ell \mid  \ell \text{ is a line in } \R^2 \}$, then
	$\dQ$ is a metric in $\c{L}$.
\end{theorem}

\begin{proof}
	The function $\dQ(\cdot,\cdot)$ is symmetric and by mapping to $\R^n$ satisfies the triangle inequality, and $\ell_1=\ell_2$ implies $\dQ(\ell_1,\ell_2)=0$; we now show if $\dQ(\ell_1,\ell_2)=0$, then $\ell_1=\ell_2$.
	
	Without loss of generality, we assume $(x_1,y_1),$ $(x_2,y_2),$ $(x_3,y_3) \in Q$ are not on the same line, which implies
	\begin{small}
		\begin{equation}   \label{three points are not on the same line}
		\left|\begin{array}{ccc}
		x_1 &  y_1    & 1 \\
		x_2 &  y_2   & 1\\
		x_3 &  y_3 & 1
		\end{array}\right| \neq 0.
		\end{equation}
	\end{small}
	Suppose $\ell_1$ and $\ell_2$ are expressed in the form:
	\[ 
	\begin{split}
	\ell_1=&\ \{ (x,y)\in \R \ |\ u_{1}^{(1)}x+u_{2}^{(1)}y+u_{3}^{(1)}=0 \},  \\
	\ell_2=&\ \{ (x,y)\in \R \ |\ u_{1}^{(2)}x+u_{2}^{(2)}y+u_{3}^{(2)}=0 \},
	\end{split}
	\] 
	where $(u_{1}^{(1)},u_{2}^{(1)},u_{3}^{(1)}),(u_{1}^{(2)},u_{2}^{(2)},u_{3}^{(2)})\in \mathbb{U}^3$ represent lines $\ell_1$ and $\ell_2$, respectively.
	If $\dQ(\ell_1,\ell_2)=0$, then we have
	\[
	x_i(u_{1}^{(1)}-u_1^{(2)})+y_i(u_2^{(1)}-u_2^{(2)})+(u_3^{(1)}-u_3^{(2)})=0
	\]
	for $i=1,2,3$. We can write this as the system
	\begin{small}
		\[ 
		\left[\begin{array}{ccc}
		x_1 &  y_1    & 1 \\
		x_2 &  y_2   & 1\\
		x_3 &  y_3 & 1
		\end{array}\right]
		\left[\begin{array}{c}
		u_{1}^{(1)}-u_{1}^{(2)} \\
		u_{2}^{(1)}-u_{2}^{(2)}\\
		u_{3}^{(1)}-u_{3}^{(2)}
		\end{array}\right]
		= 0.
		\] 
	\end{small}
	Using \eqref{three points are not on the same line}, we know it has the unique solution $[ u_{1}^{(1)}-u_{1}^{(2)}, u_{2}^{(1)}-u_{2}^{(2)}, u_{3}^{(1)}-u_{3}^{(2)}] ^T=[0,0,0]^T$.
	So, we have $ u_{1}^{(1)}=u_{1}^{(2)}$, $u_{2}^{(1)}=u_{2}^{(2)}$ and $u_{3}^{(1)}=_{3}^{(2)}$, and thus $\ell_1=\ell_2$.
\end{proof}

\paragraph*{Remark.}
In the above formulation, the absolute value $|v_{Q_i}(\ell)|$ is the distance from $(x_i,y_i)$ to the line $\ell$, i.e. $|v_{Q_i}(\ell)|=\min_{(x,y)\in \ell}((x-x_i)^2+(y-y_i)^2)^\frac{1}{2}$. Moreover, if $\ell$ is parallel to $\ell'$, then $|v_{Q_i}(\ell)-v_{Q_i}(\ell')|=\min_{(x,y)\in \ell,(x',y')\in \ell'}((x-x')^2+(y-y')^2)^\frac{1}{2}$ for any $i\in [n]$, which means $\dQ$ is a generalization of the natural offset distance between two parallel lines.

\paragraph*{Remark.}
There are several other nicely defined variants of this distance.
For a line $\ell$ we could define $\bar{v}_{Q_i}(\ell) = |v_{Q_i}(\ell)|$, as the \emph{unsigned} distance from $q_i \in Q$ to the line $\ell$.  When we consider the distance from $q_i$ to some bounded object (e.g., a trajectory in place of $\ell$), this distance is more natural.  
%
%
We are able to show in Appendix \ref{sec:metric} that under similar mild restrictions on $Q$ that this is a metric;  
the condition requires $5$ points instead of $3$.  However, we are not able to show constant-size VC-dimension for its metric balls (as we do for $\dQ$ in Section \ref{sec:VC}).  There we also introduce another matrix Frobenius norm variant.

\subsection{Distance Between Hyperplanes}
\label{distance 2D}

Now let $\c{H}=\{h \mid  h \text{ is a hyperplane  in } \R^d \}$ represent the space of all hyperplanes.
Suppose $Q=\{q_1,q_2,\cdots,q_n\}\subset \R^d$, where $q_i$ has the coordinate $(x_{i,1},x_{i,2}.\cdots,x_{i,d})$. Any hyperplane $h \in \c{H}$ can be uniquely expressed in the form
\[
h=\big\{x=(x_1,\cdots,x_d)\in \R^d \ |\ \sum\nolimits_{j=1}^d u_jx_j+u_{d+1}=0 \big\},
\]
where $(u_1,\cdots,u_{d+1})$ is a vector in
$\mathbb{U}^{d+1} :=\{u=(u_1,$
$\cdots,u_{d+1})\in \R^{d+1}  \mid \sum_{j=1}^d u_j^2=1 $ $\text{ and the first }$ nonzero entry of $u \text{ is positive}\}$,
i.e. $(u_1,\cdots,u_d)$ is the unit normal vector of $h$, and $u_{d+1}$ is the offset.
We introduce the notation
$v_Q(h)=(v_{Q_1}(h),\cdots,v_{Q_n}(h))$ where $v_{Q_i}(h)$ is again the signed distance from $q_i$ to the closest point on $h$.   We can specify
$v_{Q_i}(h)=\sum_{j=1}^d u_jx_{i,j}+u_{d+1}$,
which is a dot-product with the unit normal of $h$, plus offset $u_{d+1}$.  Now for two hyperplanes $h_1,h_2$ in $\R^d$ define
\begin{align*}
\dQ(h_1,h_2) &:=\big\|\frac{1}{\sqrt{n}}(v_Q(h_1)-v_Q(h_2))\big\|
\\ &= \Big(\sum_{i=1}^n\frac{1}{n}(v_{Q_i}(h_1)-v_{Q_i}(h_2))^2\Big)^{\frac{1}{2}}.
\end{align*}

For $Q \subset \R^d$, similar to $\dQ$ in $\R^2$, we want to consider the case that there are $d+1$ points in $Q$ which are not on the same hyperplane.  We refer to such a point set $Q$ as \emph{full rank} since if we treat the points as rows, and stack them to form a matrix, then that matrix is full rank.
Like lines in $\R^2$, a hyperplane can also be mapped to a point in $\R^n$, and if $Q$ is full rank, then no two hyperplanes will be mapped to the same point in $\R^n$.
So, similar to Theorem \ref{well-difined metric}, we can prove $\dQ$ is a metric in $\c{H}$.

\begin{theorem}\label{well-difined metric for hyperplanes}
	If $Q=\{q_1,q_2,\cdots,q_n\}\subset \R^d$ is full rank, then 
	$\dQ$
	is a metric in $\c{H}$.
\end{theorem}

\paragraph*{Remark.}
The distance can be generalized to weighted point sets and continuous probability distributions.  Suppose $Q=\{q_1,\cdots,q_n\}$ $\subset \R^d$, $W=\{w_1,\cdots,w_n\}\subset(0,\infty)$, and $\mu$ is a probability measure on $\R^d$.
For two hyperplanes $h_1,h_2$ in $\R^d$,  we define
\begin{small}
	\begin{align*}         \label{def of d_QW}
	\dQW(h_1,h_2)
	&=\Big(\sum_{i=1}^n w_i(v_{Q_i}(h_1)-v_{Q_i}(h_2))^2\Big)^{\frac{1}{2}},\\
	\mathtt{d}_\mu(h_1,h_2)
	&=\Big(\int _{\R^d} (v_{x}(h_1)-v_{x}(h_2))^2 \dir \mu(x)\ \Big)^{\frac{1}{2}},
	\end{align*}
	where $v_{x}(\cdot)$ is defined in the same way as $v_{Q_i}(\cdot)$ for $x\in \R^d$.
\end{small}

\subsection{VC-Dimension of Metric Balls for $\dQ$}
\label{sec:VC}

The distance $\dQ$ can induce a range space $(\c{H},\c{R}_Q)$, where again $\c{H}$ is the collection of all hyperplanes in $\R^d$, and
$\c{R}_Q=\{B_Q(h,r)\ | \ h\in \c{H}, r\geq 0\}$ with metric ball $B_Q(h,r)=\{h'\in\c{H}\ |\ \dQ(h,h')\leq r\}$. We prove that the VC dimension~\cite{VC71} of this range space only depends on $d$, and is independent of the number of points in $Q$.

\begin{theorem}\label{the bound of VC dimension}
	Suppose $Q\subset \R^d$ is full rank, then the VC-dimension of the range space $(\c{H},\c{R}_Q)$ is at most $\frac{1}{2}(d^2+5d+6)$.
\end{theorem}

\begin{proof}
	For any $B_Q(h_0,r) \in \c{R}_Q$, suppose $Q=\{x_1,\cdots,x_n\}$ with $x_i=(x_{i,1},\cdots,x_{i,d})$ and $h\in B_Q(h_0,r)$.
	This implies
	$\dQ(h,h_0)$ $\leq r$, so if $h$ is represented by a unique vector $(u_1,\cdots,u_{d+1})\in \mathbb{U}^{d+1}$, then we have
	\begin{small}
		\begin{equation}   \label{expression of a_i and c}
		\sum_{i=1}^n\frac{1}{n}\Big( \sum_{j=1}^d u_j x_{i,j}+u_{d+1} -v_{Q_i}(h_0) \Big)^2\leq r^2.
		\end{equation}
	\end{small}
	Since this can be viewed as a polynomial of $u_1,\cdots,u_{d+1}$, we can use a standard lifting map to convert it to a linear equation about new variables, and then use the VC-dimension of the collection of halfspaces to prove the result.
	
	To this end, we introduce the following data parameters $a_j$ [for $0 \leq j \leq d+1$] and $a_{j,j'}$ [for $1 \leq j \leq j' \leq d+1$] which only depend on $Q$, $h_0$, and $r$.  That is these only depend on the metric $\dQ$ and the choice of metric ball.
	
	\begin{small}
		\[
		\begin{split}
		&
		a_{0}=\sum_{i=1}^nv_{Q_i}(h_0)^2-nr^2, \;\;\;
		a_{d+1}=-2\sum_{i=1}^n v_{Q_i}(h_0),  \;\;\;\\
		&a_j=-2\sum_{i=1}^n x_{i,j}v_{Q_i}(h_0) \; [\text{for } 1\leq j\leq d],
		\\ &
		a_{d+1,d+1}=n,  \;\;\;\;
		a_{j,d+1}=2\sum_{i=1}^n x_{i,j} \text{ [ for } 1\leq j\leq d],
		\\&
		a_{j,j}=\sum_{i=1}^nx_{i,j}^2 \text{ [for } 1\leq j\leq d],   \;\;\; \text{ and } \;\;\;\\
		&a_{j,j'}=2\sum_{i=1}^n x_{i,j}x_{i,j'} \text{[for } 1\leq j<j'\leq d].
		\end{split}
		\]
	\end{small}
	We also introduce another set of new variables $y_j$ [for $1 \leq j \leq d+1$] and $y_{j,j'}$ [for $1 \leq j\leq j' \leq d+1$] which only depend on the choice of $h$:
	\begin{small}
		\[
		\begin{split}
		&y_j=u_j \text{ [for } 1\leq j\leq d+1]
		\;\;\; \text{ and } \;\;\;\\
		&y_{j,j'}=u_ju_{j'} \text{ [for } 1\leq j\leq j'\leq d+1].
		\end{split}
		\]
	\end{small}
	Now \eqref{expression of a_i and c} can be further rewritten as
	\begin{small}
		\[
		\sum_{j=1}^{d+1}a_jy_j+ \sum_{1\leq j\leq j'\leq d+1}a_{j,j'}y_{j,j'}+a_0\leq0.
		\]
	\end{small}
Since the $a_j$ and $a_{j,j'}$ only depend on $\dQ$, $h_0$, and $r$, and the above equation holds for any $y_j$ and $y_{j,j'}$ implied by an $h \in B_Q(h_0, r)$, then it converts $B_Q(h_0, r)$ into a halfspace in $\R^{d'}$ where
	$d'=2(d+1)+ {d+1 \choose 2}=\frac{1}{2}(d^2+5d+4)$.
	Since the VC-dimension of  halfspaces in $\R^{d'}$ is $d'+1$, the VC dimension of $(\c{H},\c{R}_Q)$ is at most $d'+1=\frac{1}{2}(d^2+5d+6)$.
\end{proof}

\paragraph*{Remark.}
This distance, metric property, and VC-dimension result extend to operate between any objects, such as polynomial models of regression, when linearized to hyperplanes in $\R^d$.

\subsection{Applications in Analysis}
\label{sec:app}

The new distance $\dQ$ for hyperplanes has many applications in statistical and algorithmic data analysis where hyperplanes map to linear models.  
%
%
For instance,
given a large varieties of regression models $H = \{h_1, h_2,$ $\ldots, h_m\}$ (e.g., stemming from different algorithms or model parameters) we can define a Gaussian-like kernel $K(h_1,h_2) = \exp(-\dQ(h_1, h_2)^2)$ and kernel density estimate $\kde_H(h) = \frac{1}{|H|} \sum_{h_i \in H} K(h,h_i)$.  The VC-dimension $\nu$ of the metric balls of $\dQ$ implies numerous stability and approximation properties of the KDE.  For instance, given a sample $S \subset H$ of size $O(\nu/\eps^2)$ ensures that with constant probability $\|\kde_H - \kde_S\|_\infty \leq \eps$~\cite{JoshiKommarajuPhillips2011}.
The embedding implies we can use Llloyd's algorithm for $k$-means clustering, and the metric property implies Gonzalez algorithm~\cite{Gon85} for $k$-center clustering will give a $2$-approximation.

\subsection{Direct Extension (literally) to Trajectories}

In this section, we show how $\dQ$ can be simply generalized to the distance between two piecewise-linear curves,
while retaining the many nice properties described above.
Let $\Gamma_k=\{\gamma \mid  \gamma \text{ is a curve in } \R^2$
$\text{ defined by $k$ ordered line segments}\}$ represent the space of all $k$-piecewise linear curves.

For any curve $\gamma \in \Gamma_k$, let its $k$ segments be $\langle s_1, s_2, \ldots, s_k \rangle$, and let these map to $k$ lines $\ell_1, \ldots, \ell_k$ where each $\ell_j$ contains $s_j$ (literally an ``extension'' of $s_j$ to a line $\ell_j$, $\ddot\smile$).  Next add two more lines: $\ell_0$ which is perpendicular to $\ell_1$ and passes through the first end point of $s_1$, and $\ell_{k+1}$ which is perpendicular to $\ell_k$ and passes through the last end point of $s_k$ (in high dimensions, some canonical choice of $\ell_0$ and $\ell_{k+1}$ is needed).  We now represent $\gamma$ as the ordered set of $k+2$ lines $(\ell_0, \ell_1, \ldots, \ell_k, \ell_{k+1})$.
This mapping is $1$to$1$, since segments $s_i$ and $s_{i+1}$ share a common end point, and this defines the intersection between $\ell_i$ and $\ell_{i+1}$.  The intersections with the added lines $\ell_0$ and $\ell_{k+1}$ define first and last endpoints of $s_1$ and $s_k$, and these endpoints are sufficient to define $\gamma$.

Now for two curves $\gamma^{(1)},\gamma^{(2)} \in \Gamma_k$, we define the distance using their line representations $(\ell_0^{(1)}, \ldots, \ell_{k+1}^{(1)})$ and $(\ell_0^{(2)}, \ldots, \ell_{k+1}^{(2)})$, respectively, as
\[
\dQk(\gamma^{(1)},\gamma^{(2)}):=\frac{1}{k+2}\Big(\sum\nolimits_{i=0}^{k+1} \dQ \big(\ell_i^{(1)},\ell_i^{(2)}\big)\Big).
\]

\paragraph*{Metric.}
If $\dQk(\gamma^{(1)},\gamma^{(2)})=0$, then $\dQ\big(\ell_i^{(1)},\ell_i^{(2)}\big)=0$ for all $i\in[k]$, which implies $\ell_i^{(1)}=\ell_i^{(2)}$ if $Q$ is full rank.  Combined with the $1$to$1$ nature of the mapping from $\gamma = (s_1, \ldots, s_k)$ to $(\ell_0, \ldots, \ell_{k+1})$, we have that if $Q$ is full rank, then $\dQk$ is a metric over $\Gamma_k$.

\paragraph*{VC dimension.}
The distance $\dQk(\cdot,\cdot)$ can induce a range space $(\Gamma_k,\c{S}_{Q,k})$, where again $\Gamma_k$ is the collection of all $k$-piecewise linear curves in $\R^2$, and
$\c{S}_{Q,k}=\{B_Q(\gamma,r)\ | \ \gamma\in \Gamma_k, r\geq 0\}$ with metric ball $B_Q(\gamma,r)=\{\gamma'\in \Gamma_k \ |\ \dQk(\gamma,\gamma')\leq r\}$. Using the straightforward extensions of the method in the proof of Theorem \ref{the bound of VC dimension}, we can show the VC dimension of this range space only depends on $k$, and is independent of the number of points in $Q$.
Specifically, for full rank $Q \subset \R^2$, the VC-dimension of $(\Gamma_k, \c{S}_{Q,k})$ is at most $9k + 19$.

While retaining all above mathematical properties, this distance is unintuitive, and as we show in Section \ref{sec:traj-analysis}, can perform less than optimally.  We next develop other trajectory distances which are more intuitive, but have weaker mathematical properties.


\section{Landmark Distances Between Trajectories}
\label{sec:curves}

In this section, we define two variants of $\dQ$ for trajectories, focused on their modeling as piecewise-linear curves on $\R^2$.  We let $\Gamma$ define the set of such curves, and they are specified by a series of critical points $\langle c_0, c_1, \ldots, c_k \rangle$.  The curve $\gamma \in \Gamma$ is the subset of $\R^2$ defined by the $k$ segments $s_1, s_2, \ldots, s_k$ where $s_i = c_{i-1} c_i$ is the continuous set of points between critical points $c_{i-1}$ and $c_i$.
For notational convenience, we will describe all curves as having $k$ segments, but the distance will not require this.  Moreover, since we model the trajectory as a continuous subset of $\R^2$, it will not distinguish trajectories of different speeds or moving in opposite directions but following the same paths.

Now for a curve $\gamma \in \Gamma$ and size $n$ point set $Q \subset \b{R}^2$, define
$v_i = \min_{p \in \gamma} \|q_i - p\|$ and $p_i = \arg\min_{p \in \gamma} \|q_i - p\|$; see Figure \ref{fig:traj-dist}.
For two curves $\gamma^{(1)}$ and $\gamma^{(2)}$ denote these values as $v_i^{(1)}$, $p_i^{(1)}$ and $v_i^{(2)}$, $p_i^{(2)}$ respectively. Our distances are then defined as:
\begin{align*}
\dQ\big(\gamma^{(1)},\gamma^{(2)}\big) &=  \Big(\frac{1}{n}\sum_{i=1}^{n}\big(v_i^{(1)}-v_i^{(2)}\big)^2\Big)^{\frac{1}{2}},
\\
\dQ^{\pi}\big(\gamma^{(1)},\gamma^{(2)}\big) & =\frac{1}{n}\sum_{i=1}^{n}\big(\|p_i^{(1)}-p_i^{(2)}\|\big).
\end{align*}
The standard variant $\dQ$ is the analog of the version for halfspaces, where as the second variant $\dQ^\pi$ (the \emph{projected landmark distance}) projects $Q$ onto the closest points of the curves, and then computes the average distances with respect to these projected points.

\begin{figure}[t]
	\vspace{-2mm}
	\begin{center}
		\includegraphics[width=.8\linewidth]{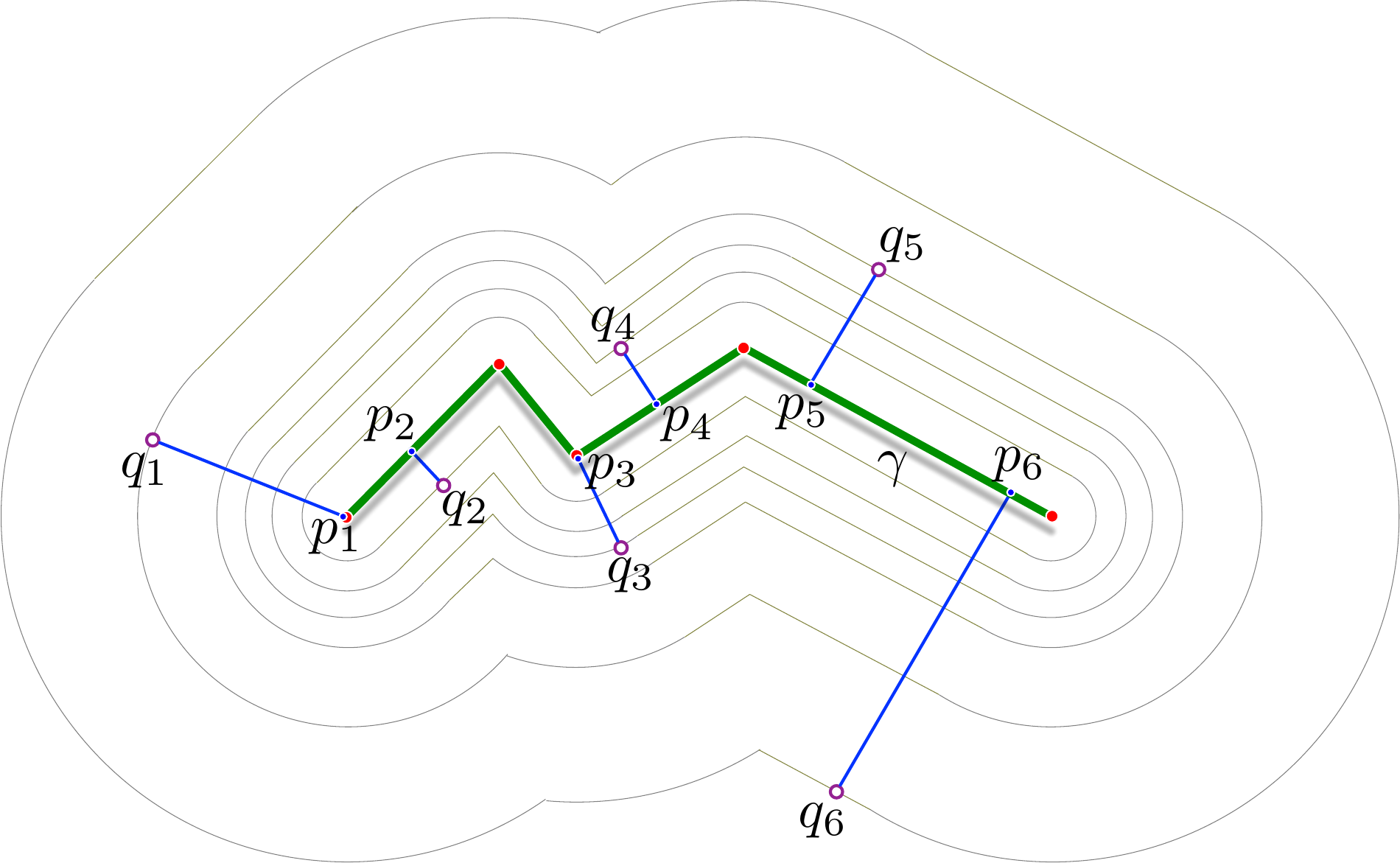}
	\end{center}
	\vspace{-3mm}
	\caption{Illustrating $q_i$ and $p_i$ on a trajectory for $\dQ$ and $\dQ^\pi$.\label{fig:traj-dist}}
	\vspace{-1mm}
\end{figure}

\subsection{Metric Properties}
\label{traj-metric}

In this section, we show a reasonable condition for the trajectories and $Q$ so that both variants are metrics.  As with lines and halfspaces, these distances are always pseudometrics: the symmetry and triangle inequality are direct consequences of the embedding to Euclidean space.
The only restriction of the trajectories is to ensure that two distinct curves do not have a distance $0$, and in our arguments this requires that the critical points have some non-zero separation from other parts of the curve.   These restrictions may not be necessary, but it makes the proofs simple enough.  Then we basically just require that $Q$ is sufficiently dense; if we decide many of these points are irrelevant, we can reduce the weights on those points (keeping them non-zero) and the metric properties still hold.

We define a family of curves $\Gamma_{\tau} \subset \Gamma$ so each $\gamma \in \Gamma_{\tau}$ has two restrictions:
(R1) Each angle $\angle_{[c_{i-1}, c_i, c_{i+1}]}$ about an internal critical point $c_i$ is non-zero (i.e., in $(0, \pi)$).
(R2) Each critical point $c_i$ is \emph{$\tau$-separated}, that is the ball $B(c_i, \tau) = \{x \in \R^2 \mid \|x-c_i\| \leq \tau\}$ only intersects the two adjacent segments $s_{i-1}$ and $s_i$ of $\gamma$, or one adjacent segment for end points (i.e., only the $s_1$ for $c_0$ and $s_k$ for $c_k$, if $\gamma$ has $k$ line segments).  The $\tau$-separated property, for instance, enforces that critical points are at least a distance $\tau$ apart.

We next restrict that all curves (and $Q$) lie in a sufficiently large bounded region $\Omega \subset \R^2$.  Let $\Gamma_{\tau}(\Omega)$ be the subset of $\Gamma_{\tau}$ where all curves $\gamma$ have all critical points within $\Omega$, and in particular, no $c_i \in \gamma \in \Gamma_{\tau}(\Omega)$ is within a distance $\tau$ of the boundary of $\Omega$.
Now for $\eta > 0$, define an infinite grid $G_\eta = \{g_v \in \R^2 |\  g_v = \eta v \text{ for } v = (v_1, v_2) \in \mathbb{Z}^2 \}$, where $\mathbb{Z}$ is all integers.

\begin{figure}
	\includegraphics[width=0.46\linewidth]{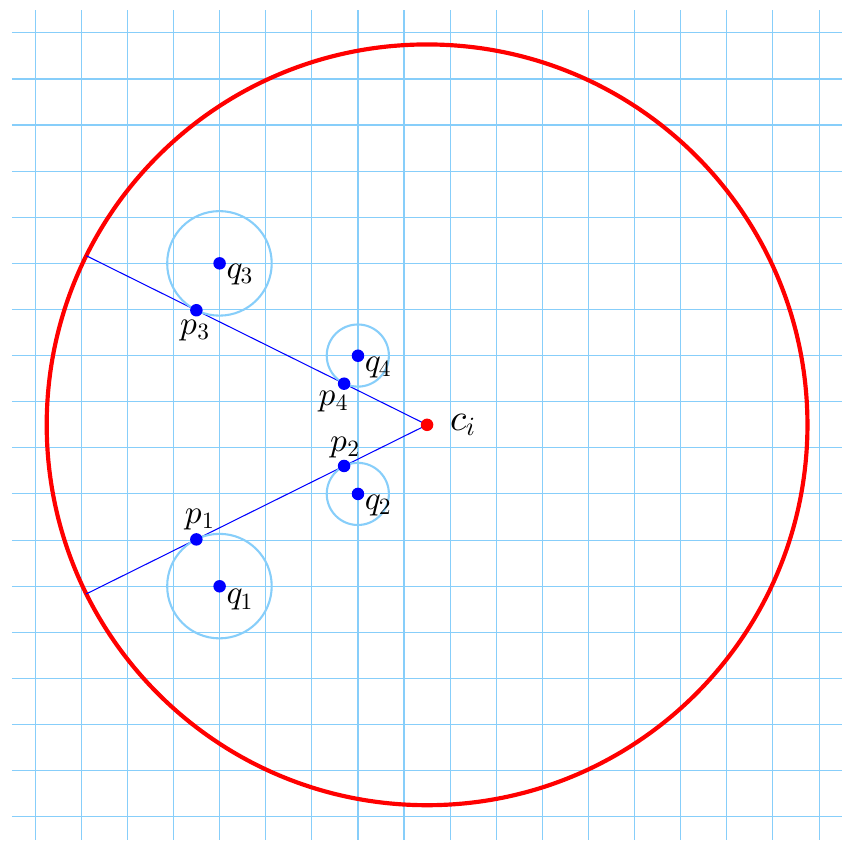}
	\vspace{-.16in}
	\caption{$c_i$ is a critical point of $\gamma^{(1)}$ } \label{metric_property}
	\vspace{-.2in}
\end{figure}

\begin{theorem}  \label{well-difined metric for dQ1 and dQ2}
	For $Q = G_\eta \cap \Omega$ and $\eta \leq \frac{\tau}{16}$, both $\dQ$ and $\dQ^\pi$ are metrics in $\Gamma_{\tau}(\Omega)$.
	%
\end{theorem}

\begin{proof}
	We prove this theorem for $\dQ^\pi$, and the proof for $\dQ$ is similar and given in Appendix \ref{metric peroperties for dQ1 and dQ2}.
	Suppose $\gamma^{(1)},\gamma^{(2)}\in \Gamma_{\tau}(\Omega)$ have critical points $c_0, c_1,...c_{k}$ and $c_0', c_1',...c_{k'}'$ respectively. We only need to show if $\dQ^\pi(\gamma^{(1)},\gamma^{(2)})=0$ then $\gamma^{(1)}=\gamma^{(2)}$.
	Here, if two piecewise-linear curves have the same critical points and their orders are the same or reverse of each other, then these two curves are regarded as the same curve.
	
The argument follows 4 steps assuming $\dQ^\pi(\gamma^{(1)},\gamma^{(2)})=0$:
 (Step 1) Around each critical point $c_i$ of $\gamma^{(1)}$, we can identify at least $4$ points $q_1, q_2, q_3, q_4$ that map to $p_1, p_2, p_3, p_4$, two each on the two segments adjacent to $c_i$.
 (Step 2) The segments between defined by $\overline{p_1 p_2}$ and $\overline{p_3 p_4}$ must also be part of $\gamma^{(2)}$.
 (Step 3) The line extension of those two line segment must intersect at $c_i$, and this must also be critical point on $\gamma^{(2)}$
 (Step 4) Because these Steps 1-3 can be repeated for all critical points on $\gamma{(1)}$ and on $\gamma^{(2)}$, they must share critical points and connecting line segments, and be the same curves.
	
	We formalize these steps based on three observations:
	(O1) For $q\in Q, \gamma\in \Gamma$, $p=\arg$ $\min_{p'\in \gamma}\|p'-q\|$, suppose
	$l$ is the tangent line of the circle $C(q,\|q-p\|)$
	where $q$ is center and $\|q-p\|$ its radius, at point $p$. If $l\cap B(p,\delta)$
	is not apart of $\gamma$  for all $\delta>0$ , then $p$ must be a critical point of $\gamma$.
	(O2) If $\gamma\in \Gamma_{\tau}$, then in any ball with radius $\frac{\tau}{2}$, there is at most one critical point of $\gamma$.
	(O3) If a point moves along $\gamma\in \Gamma$, then it can only stop or change direction at critical points.

\underline{Step 1:}	
Suppose $c_i=(x_i,y_i)$ ($1\leq i \leq k-1$) is a critical point of $\gamma^{(1)}$, and consider a ball $B(c_i,\frac{1}{2}\tau)$, as shown in Figure \ref{metric_property}.
	Since the side length of each grid cell is $\eta\leq \frac{1}{16}\tau$, from the $\tau$-separated property (R2) we know for any  $q\in Q\cap B(c_i, \frac{\tau}{2})$, $p=\arg \min_{p'\in \gamma^{(1)}}\|p'-q\|$ is in $B(c_i,\frac{\tau}{2})$.
	So, there exist two points $q_1,q_2$ that are mapped to points $p_1,p_2$ on one line segment of $\gamma^{(1)}$ and
	another two points $q_3,q_4$ are mapped to points $p_3,p_4$ on the other line segment of $\gamma^{(1)}$ in $B(c_i,\frac{\tau}{2})$.
	Since $\dQ^\pi(\gamma^{(1)},\gamma^{(2)})=0$, we know $p_1,p_2,p_3,p_4$ are also on $\gamma^{(2)}$.
	
\underline{Step 2:}	
We assert the line segment $p_1p_2$ must be a part of $\gamma^{(2)}$. From (O2), we know  $p_1$ and $p_2$ cannot both be the critical point of $\gamma^{(2)}$ at the same time, so we assume $p_1$ is not a critical point. Thus, from (O1) we know a small part of tangent line $l$
	of circle $C(q_1,\|q_1-p_1\|)$ at $p_1$ is a part of $\gamma^{(2)}$. If $p_2$ is a critical point of $\gamma^{(2)}$, then from (O3) and (O2) we know the line segment $p_1p_2$ must be a part of $\gamma^{(2)}$. If $p_2$ is not a critical point of $\gamma^{(2)}$, then from
	(O1) we know a small part of tangent line $l$ of circle $C(q_1,\|q_1-p_1\|)$ at $p_2$ is a part of $\gamma^{(2)}$. So, in this case,
	(O3) and (O2) implies the line segment $p_1p_2$ is a part of $\gamma^{(2)}$.
	Using a similar argument, we know the line segment $p_3p_4$ is also a part of $\gamma^{(2)}$.
	
\underline{Step 3:}	
We extend the line $\overline{p_1p_2}$ from $p_1$ to $p_2$ and the line $\overline{p_3p_4}$ from $p_3$ to $p_4$.
	Suppose they intersect with the boundary of $B(c_i, \frac{\tau}{2})$ at $p_2'$ and $p_4'$ respectively.
	Since $\gamma^{(2)}$ cannot go into the interior of any ball with centers in $Q\cap B(c_i,\frac{\tau}{2})$, from (O3) we know
	there must be one critical point in line segment $p_2p_2'$. For the same reason, there must be one critical point in line segment $p_4p_4'$. Thus, (O2) implies $c_i$ is a critical points of $\gamma^{(2)}$.

\underline{Step 4:}	
Considering that $\gamma^{(2)}$ has to pass through $p_1,p_2,p_3,p_4$ and $c_i$, from $\tau$-separated property (R2), we know $\gamma^{(1)}$ and $\gamma^{(2)}$ must overlap with each other in $B(c_i,\frac{\tau}{2})$.
	For two endpoints $c_0$ and $c_k$ we can make the same argument, which means in a neighborhood of each critical point of $\gamma^{(1)}$, $\gamma^{(1)}$ overlaps with $\gamma^{(2)}$. This means $\{c_0,c_1,\cdots,c_k\}$ is a subset of $\{c_0',c_1',\cdots,$ $c_{k'}'\}$. Using the same argument $\{c_0',c_1',\cdots,c_{k'}'\}$ is a subset of \{$c_0,c_1,\cdots,c_k\}$.
	Therefore, $k=k'$ and we know $\gamma^{(1)}$ and $\gamma^{(2)}$ must have the same critical points and their orders must be the same or reverse of each other.
\end{proof}

\paragraph*{Remark.}
We did not try to optimize constants.  The point is that for most families of trajectories, with $Q$ sufficiently dense our distances are metrics, not just pseudometrics.  In practice these distances will work for small sets $Q$ (see below).

\section{Trajectories Analysis via New Distances}
\label{sec:traj-analysis}
We demonstrate that  $\dQ$ and $\dQ^\pi$ (and to lesser extent $\dQ^{\leftrightarrow}$) work effectively on real world problems.  These approaches achieve state-of-the-art performance, are incredibly simple to use, and their sketched representation plugs directly into $k$-means clustering, KNN or SVM classifiers, or ANN libraries.  We show that only a small number of landmarks are needed for good accuracy, and when certain landmarks are especially meaningful, our approaches can be easily tuned to achieve very high accuracy.

\subsection{Related Trajectory Distances, and Landmarks}
\label{sec:related}
There are by now numerous definitions of trajectories, with a variety of different aspects they can model and take into account.

We compare the classification errors found using $\dQk$, $\dQ$ and $\dQ^\pi$ with a series of representative distances for trajectories.  These are:
Euclidean distance among the critical points (Eu)~\cite{ZZKH2006},
discrete Frechet distance (dF)~\cite{TEHM1994},
dynamic time warping distance (DTW)~\cite{YJF1998},
discrete Hausdorff distance (dH)~\cite{FM2008},
longest common subsequence distance (LCSS)~\cite{VKG2002},
edit distance for real sequences (EDR)~\cite{COO2005}.
We also compare against the recently proposed
locality sensitive hashing distance (LSH1$_Q$), and the ordered version of  locality sensitive hashing distance (LSH2$_Q$)~\cite{ACKGS2018}, which consider the intersection of the trajectories with a set of disks.  This is conceptually similar to our methods, where we can think of the landmarks $Q$ as the centers of disks (as we do in experiments), and their approach requires a radius parameter $r$ for all disks, and is not a metric.
The definitions of these distances are given in Appendix \ref{def of distances between traojectories}.

To find the best parameters to minimized the error, for LCSS we tested $\eps\in \{0.001,0.005,0.01,0.015,\cdots,0.055\}$, $\delta\in\{1,2,3,\cdots,$ $10\}$, and for EDR we tested $\varepsilon\in \{0.001,0.005,0.01,0.015,\cdots,$ $0.055\}$, and for LSH1$_Q$ and LSH2$_Q$ we tested $r \in \{0.005,0.01,0.02,$ $\cdots,0.11\}$.
Since in all experiments (except Section \ref{using dQ in knn}), each trajectory is represented by a sequence of 10 critical points, it is enough to take the largest value of $\delta$ as 10 for LCSS. We only show the best results in this section, but provide the results of other parameter settings in Appendix \ref{The error of LCSS, EAR and LSH with Other Parameters}.

Zhang \etal~\cite{ZZKH2006} conducted a large comparison of trajectory distances and showed that in most cases Eu is general enough, efficient, and a superior or nearly as good model as any other ; we include dF and DTW as examples which search over all possible alignments and thus do not require the same number of or aligned critical points on both curves.  The restriction that trajectories have the same number of critical points is also not required for dH, EDR, LSH1$_Q$, and LSH2$_Q$, but in comparisons we always first reduce all trajectories to $10$ critical points (with Douglas-Peucker), except in Section \ref{using dQ in knn}, so a fair comparison to all metrics can be made.

Even beyond the recent trajectory LSH paper~\cite{ACKGS2018}, the use of waypoints to provide a distance between trajectories is not new.  However, they are typically used in other contexts, such as annotating with geolocated social media~\cite{WLLWH15}.  Or for instance, in the context of a line of work~\cite{DRA2006,EKNY2005,RTJ2010} seeking to find the $k$ nearest time-encoded trajectories to a given point at a specific time, Lin \etal~\cite{DRA2006} use a set of landmarks $Q$ to map trajectories and query points into the Voronoi cells of $Q$ to quickly help in pruning.

\subsection{Warm-up: $k$-means Clustering} \label{Warm-up: k-means Clustering}
As a warm up, we consider clustering the $42$ trajectories from user $id_{155}$ in the Geolife GPS trajectory dataset~\cite{geolife-gps-trajectory-dataset-user-guide}.   We randomly choose $20$ spread-out Beijing POIs as the landmark set $Q$, shown as orange dots in Figure \ref{fig:cluster}.    Using $\dQ$, this maps each trajectory $\gamma$ to $\R^{20}$, and we directly run Lloyd's algorithm for $k$-means clustering with $k=2,3$, and color-code the corresponding trajectories in Figure \ref{fig:cluster}.  We observe that although the trajectories are intertwined, there is a central-city cluster found in both cases, and either $1$ or $2$ clusters found on the north side.

\begin{figure}[t]
	\includegraphics[width=0.49\linewidth]{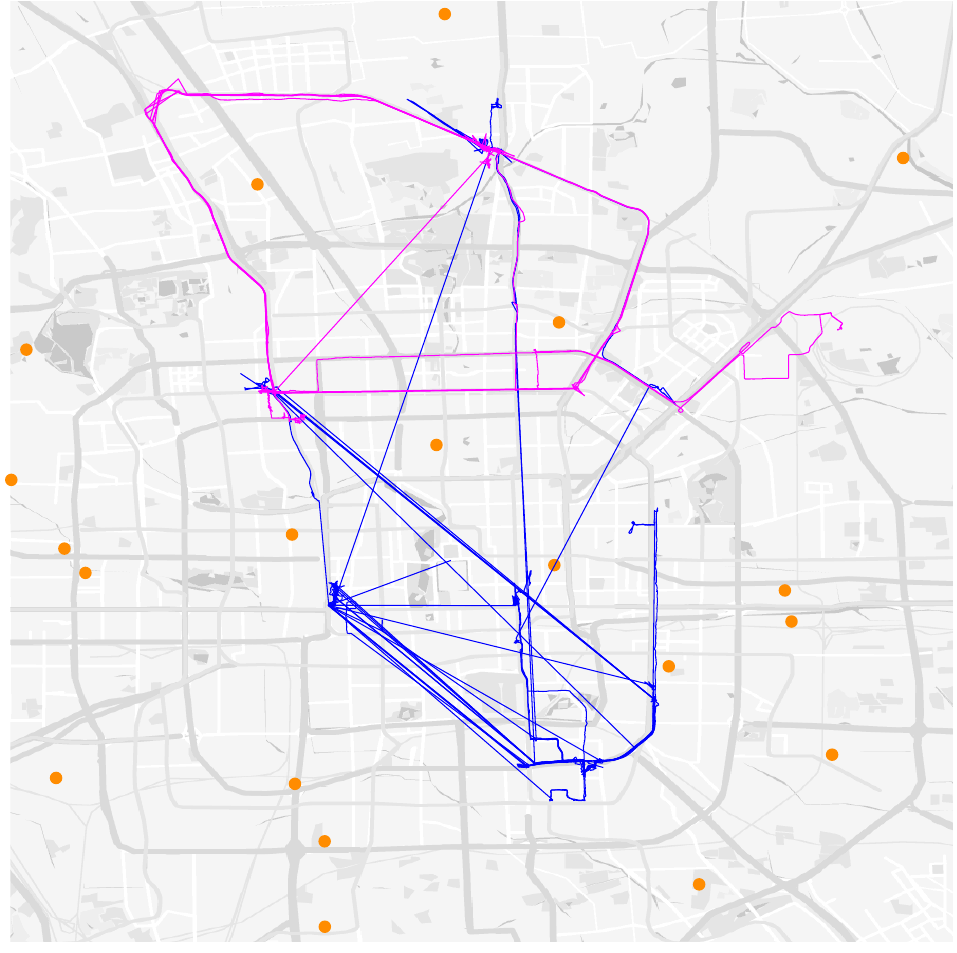}
	\includegraphics[width=0.49\linewidth]{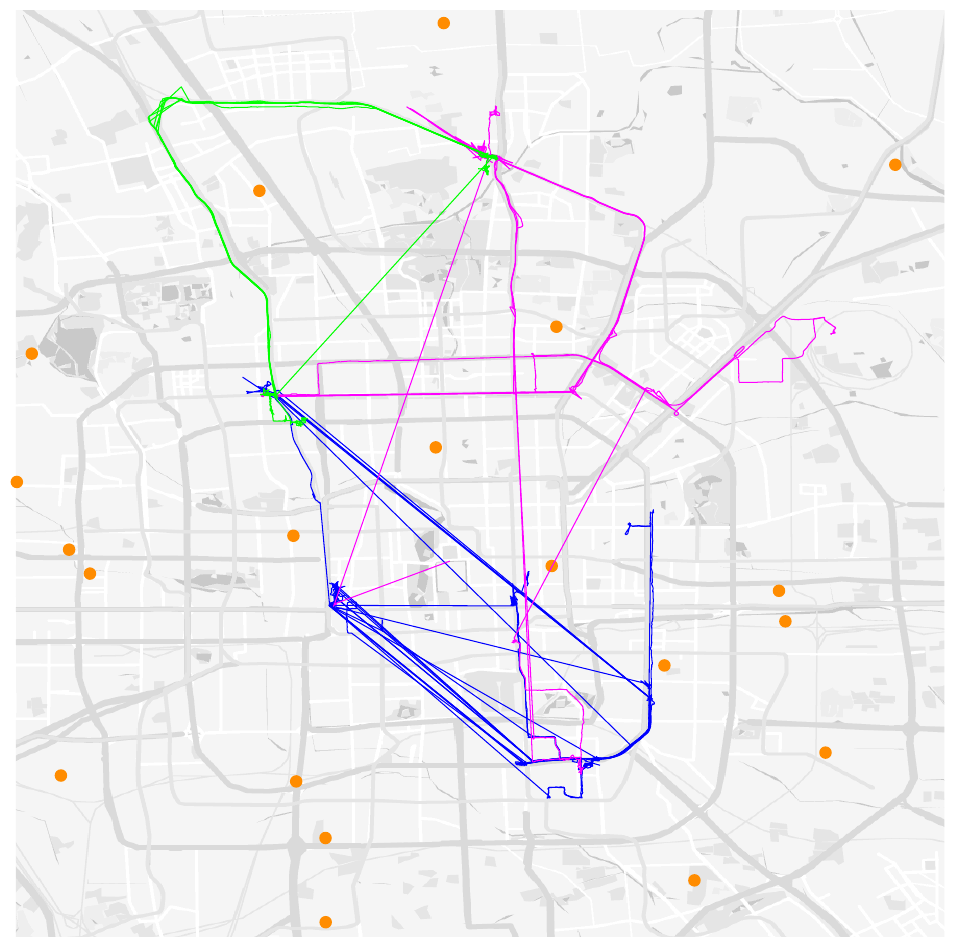}
	\vspace{-.15in}
	\caption{$2$ or $3$ clusters (color-coded) under $k$-means on $\dQ$ with $20$ landmarks $Q$ shown overlaid on Beijing.\label{fig:cluster}}
\end{figure}

\begin{table*}[t]
	\center
	\caption{Classification error on Beijing Drivers with KNN.}
	\label{table KNN real data1}
	\vspace{-.1in}
	
	\begin{tabular}{r|ccccccccccc}
		\hline
		&$\dQk$ & $\dQ$ & $\dQ^\pi$ & Eu & dF & DTW & dH  &  LCSS & EDR & LSH1$_Q$ & LSH2$_Q$\\
		best param & - & - & - & - & - &- &- & $\eps$=0.005,$\delta$=10 & $\eps$=0.005 & $r$=0.06 & $r$=0.1 \\
		\hline
		mean     & 0.1703 & 0.0817 & 0.0724 & 0.0811 & 0.1045 & 0.0722 & 0.0883  & 0.0714 & 0.0802 & 0.1290 &   0.2409\\
		median   & 0.1458 & 0.0667 & 0.0581 & 0.0654 & 0.0873 & 0.0571 & 0.0722  & 0.0500 & 0.0554 & 0.0949 &   0.2182 \\
		variance & 0.0108 & 0.0039 & 0.0033 & 0.0040 & 0.0054 & 0.0036 & 0.0043  & 0.0054 & 0.0070 & 0.0128 &   0.0210 \\
		\hline
		
	\end{tabular}
\end{table*}

\begin{table}[b]
	\center
	\caption{Classification error on Beijing Drivers with SVM.}  \label{table SVM_real data1}
	\vspace{-.1in}
	\begin{tabular}{cr|cccc}
		\hline
		kernel & statistics& $\dQk$ &$\dQ$ & $\dQ^\pi$ & Eu \\
		\hline
		\multirow{3}{*}{linear}    &mean     &  0.2170&  0.2066&  0.2046&  0.2173\\
		                           &median   &  0.1987&  0.1851&  0.1892&  0.2000\\
	              	               &variance &  0.0140&  0.0158&  0.0149&  0.0164\\
		\hline
		\multirow{3}{*}{quadratic} &mean     &  0.2327&  0.2190&  0.2000&  0.2377\\
		                           &median   &  0.2000&  0.1778&  0.1455&  0.1949\\
		                           &variance &  0.0200&  0.0281&  0.0284&  0.0278\\
		\hline
		\multirow{3}{*}{Gaussian}  &mean     &  0.1725&  0.0727&  0.0733&  0.0845\\
	                          	   &median   &  0.1509&  0.0587&  0.0588&  0.0690\\
		                           &variance &  0.0110&  0.0035&  0.0036&  0.0045\\
		\hline
	\end{tabular}

\end{table}

\subsection{Classifying Trajectories 1: Beijing Drivers} \label{Classifying Trajectories 1: Beijing Drivers}
We also consider classifying trajectories from users in the Geolife dataset~\cite{geolife-gps-trajectory-dataset-user-guide} with the same $20$ POI landmarks $Q$ as in the clustering example.
There are 182 users, and each user has several trajectories in Beijing. We only consider those trajectories with more than
10 critical points, and if a user has less than 10 such trajectories, then we remove this user. Thus, 54 users are removed, and in the remaining 128 users, 20 of them have more than 200 trajectories. For each of these users, we just randomly sample 200 trajectories (without replacement), to avoid severe imbalance in classification -- dealing with the imbalance challenge is not the focus of this paper.

Suppose two users with  $id_1$ and $id_2$ have two sets of trajectories $\Gamma^{(1)}$ and $\Gamma^{(2)}$ respectively.
Letting $|\Gamma^{(1)}|=m_1$ and $|\Gamma^{(2)}|=m_2$, we randomly sample $\lfloor \frac{3m_1}{10}\rfloor$ trajectories from $\Gamma^{(1)}$ and $\lfloor \frac{3m_2}{10}\rfloor$ trajectories from $\Gamma^{(2)}$ respectively to form a test set, and use the other trajectories in $\Gamma^{(1)}\cup \Gamma^{(2)}$ as the training data. Then we choose an algorithm and metric to do classification, and compute the error. For users with  $id_1$ and $id_2$, we do this 10 times and take the mean error as $\text{error}(id_1,id_2)$. We compute $\text{error}(id_1,id_2)$ for all 8128 pairs of 128 uses, and then output the mean, median, and variance of these  8128 errors.

For all of these 10 distances, we use the KNN classification ($K=5$); see Table \ref{table KNN real data1}.  The lowest error rates of about $7\%$ error is achieved by $\dQ^\pi$, DTW and LCSS; they are within the variance bounds of each other.  Then $\dQ$, Eu, and EDR achieve error about $8\%$, again within the error bounds of each other.  Other metrics perform worse with for example, dF at $10\%$, LSH1$_Q$ at $13\%$, $\dQk$ at $17\%$, and LSH2$_Q$ at $24\%$ error.

For  $\dQ$, $\dQ^\pi$ and Eu, since they map a trajectory to a vector in Euclidean space, we can also directly use SVM to classify these vectors.  We use fitcsvm in matlab R2018b and set `IterationLimit' (the maximum iteration number) as 200{,}000 for all kernel functions, and set `KernelScale' as `auto' for Gaussian kernel.
From Table \ref{table SVM_real data1}, we can see for SVM with three kinds of kernel functions, both $\dQ$ and $\dQ^\pi$ are better than Eu.  In the case of Gaussian SVM, both $\dQ$ and $\dQ^\pi$ achieve an error rate of about $7\%$ which is less than the about $8\%$ achieved by Eu, and this difference is larger than the variances.
Again in this SVM setting $\dQk$ performs much worse (for Gaussian kernels) or comparable to other measures, about the same as Eu, (linear quadratic kernels).

\begin{table}[b]
	\vspace{-.1in}\center
	\caption{Classification error on Beijing with $|Q|=200$.\label{tbl:mean200}}
	\vspace{-.1in}
	\begin{tabular}{C{0.3cm}R{1cm}|C{0.6cm}ccc}
		\hline
		& statistics & KNN & linear-SVM & quad-SVM & Gauss-SVM   \\
		\hline
		\multirow{3}{*}{$\dQ$}      & mean     & 0.0801 & 0.1419 & 0.1398 & 0.0722   \\
	                  	            & median   & 0.0650 & 0.1125 & 0.0909 & 0.0581   \\
		                            & variance & 0.0038 & 0.0112 & 0.0203 & 0.0035   \\
		\hline
		\multirow{3}{*}{ $\dQ^\pi$} &mean      & \pmb{0.0708} & 0.1432 & 0.2606 & 0.0726  \\
	                        	    &median    & 0.0558 & 0.1179 & 0.2222 & 0.0583  \\
		                            &variance  & 0.0033 & 0.0104 & 0.0373 & 0.0036  \\
		\hline
		\multirow{3}{*}{$\dQk$}     & mean     & 0.1711&  0.2362&   0.2673&    0.1735\\
	   	                            & median   & 0.1471&  0.2200&   0.2460&    0.1529\\
		                            & variance & 0.0108&  0.0138&   0.0212&    0.0110\\
		\hline
	\end{tabular}

\end{table}

As we increase the size of $Q$ to 200 (chosen at random), then both $\dQ$ and $\dQ^\pi$ slightly improve in performance, but not drastically, and $\dQk$ performs about the same.  The error statistics is shown in Table \ref{tbl:mean200}, from which we can see for KNN, the performance of $\dQ$ is better than Euclidean distance, and $\dQ^\pi$  provides the smallest error (mean error $0.0708$, smaller than $0.0714$ of LCSS).
Moreover, we can see as $|Q|$ increases, the error of $\dQ$ and $\dQ^\pi$ with three kernel functions all decrease, except $\dQ^\pi$ with quadratic kernel. When we use quadratic kernel, the algorithm takes a long time to converge, and for $|Q|=200$, the dimension of vectors used in $\dQ^\pi$ is 400, so the algorithm may not converge within 200000 iterations. The relatively small improvement also demonstrates that even with a small size, random $Q$, the distances still perform at or near the state-of-the-art.

\begin{table}[b]
	\center
	\caption{Classification error on Bus vs. Car.}
	\label{table KNN and SVM real data2 Q1}
	\vspace{-.1in}
	\begin{tabular}
		{L{1cm}R{2.7cm}|ccl}
		\hline
		&distance &  mean  &  median  & variance \\
		\hline
	\multirow{16}{*}{KNN}	&$\dQki$    & 0.2027  & 0.1944 & 0.0042 \\
	    &$\dQkii$   & 0.2148  & 0.2222 & 0.0039 \\
	    &$\dQi$     & 0.2331  & 0.2222 &  0.0045\\
	    &$\dQii$    & 0.2229  &  0.2222 & 0.0041 \\
	    &$\dQi^\pi$ & 0.2608  &  0.2500  & 0.0039 \\
		&$\dQii^\pi$   &  0.2505 & 0.2500  & 0.0039 \\
		&Eu & 0.3323 & 0.3333   & 0.0044 \\
		&dF &  0.3431  & 0.3333 & 0.0045 \\
		&DTW & 0.3118 & 0.3056 & 0.0046 \\
		&dH & 0.3284 & 0.3333 & 0.0039 \\
		&LCSS ($\eps$=0.015,$\delta$=3) &  0.2448 &  0.2500 &  0.0037 \\
		&EDR ( $\eps$=0.015) &  0.2640 & 0.2500 & 0.0039 \\
		&LSH1$_{Q_1}$ ($r$=0.02) & 0.2673 & 0.2778 & 0.0020 \\
		&LSH2$_{Q_1}$ ($r$=0.08) & 0.2516 &  0.2500 & 0.0022 \\
		&LSH1$_{Q_2}$ ($r$=0.03) & 0.2209 & 0.2222 &  0.0039 \\
		&LSH2$_{Q_2}$ ($r$=0.05) & 0.2690 & 0.2778  & 0.0022 \\
		\hline
		\multirow{7}{*}{linear SVM}
		&Eu & 0.3624 & 0.3611 &  0.00007 \\
		&$\dQki$ & 0.3652 & 0.3611 & 0.0002 \\
		&$\dQkii$ &  0.3655 & 0.3611 & 0.0002 \\
		&$\dQi$  & 0.3611 & 0.3611 & 0  \\
		&$\dQii$ & 0.3611  & 0.3611 & 0   \\
		&$\dQi^\pi$ &  0.3611 & 0.3611 & 0 \\
		&$\dQii^\pi$ & 0.3612 & 0.3611 & 0.000003 \\
		\hline
		\multirow{7}{*}{quadratic SVM}
		&Eu & 0.3609 & 0.3611 & 0.0044\\
		&$\dQki$ &  0.3645 & 0.3611 & 0.0004  \\
		&$\dQkii$ & 0.3140 & 0.3056 & 0.0017 \\
		&$\dQi$  & 0.3617 & 0.3611 &  0.00003\\
		&$\dQii$ & 0.3625 & 0.3611 & 0.00008\\
		&$\dQi^\pi$ & 0.2644 &  0.2500 & 0.0042\\
		&$\dQii^\pi$ &  0.2828 & 0.2778 &  0.0045\\
		\hline
		\multirow{7}{*}{Gaussian SVM}
		&Eu &  0.2239 & 0.2222 & 0.0034\\
		&$\dQki$ &  0.1940  & 0.1944 & 0.0031 \\
		&$\dQkii$ & 0.2120 & 0.2222 & 0.0032\\
		&$\dQi$  & 0.1894 &  0.1944 & 0.0029 \\
		&$\dQii$ &  0.1968  & 0.1944 & 0.0033\\
		&$\dQi^\pi$ & \pmb{0.1659} & 0.1667 &  0.0033 \\
		&$\dQii^\pi$ & 0.1731 &  0.1667  & 0.0033 \\
		\hline
	\end{tabular}
\end{table}

\subsection{Classifying Trajectories 2: Bus versus Car}
\label{Classifying Trajectories 2: Bus versus Car}

As another example, we consider the GPS Trajectories Data Set ~\cite{GTDS2016} in UCI machine learning repository.
There are 87 car trajectories, and 76 bus trajectories in Aracaju, a city of Brazil. We remove those trajectories having less than 10 critical points, and then 78 car trajectories and 45 bus trajectories are left.
For these 123 trajectories
are shown in Figure \ref{data_synthetic and real data 2}(Left), where pink curves are car trajectories and blue curves are bus trajectories. We hand-pick 10 points as $Q_1$ such that each point is close to one class of trajectories, and randomly generate 20 points as $Q_2$. Each time we randomly choose 23 car trajectories and 13 bus trajectories as test data, and use other trajectories as training data to perform classification experiments, and compute the error. We do this 1000 times and then compute the mean, median and variance of the error for each algorithm.

The results are shown in Table \ref	{table KNN and SVM real data2 Q1}, and we see the KNN classification results using all $14$ distance, using either $Q_1$ (10 chosen near data) or $Q_2$ (20 randomly chosen).  The results are slightly better for $Q_2$ in almost all distances $\dQ$, $\dQ^\pi$, LSH$1_Q$, and LSH$2_Q$ -- except $\dQk$.  In these experiments on $Q_2$, the best mean error (about $21\%$ to $22\%$) is achieved by $\dQk$, $\dQ$, and LSH$1_Q$ (which required a parameter search).  The best error is about $20\%$ by $\dQk$ using $Q_2$.
While $\dQ^\pi$, LCSS, EDR, and LSH$2_Q$ achieve error between $25\%$ and $27\%$.  Noticeably, the methods which were competitive with $\dQ$ and $\dQ^\pi$ on the Beijing Drivers data are EDR, which required a parameter tuned, as well as DTW and Eu, which now have error rate above $31\%$.  As a baseline, always predicting ``car" obtains $36\%$ error.

We show the results of applying SVM in Table \ref{table KNN and SVM real data2 Q1}. Again the difference is small between $Q_2$ and $Q_1$.  And while the linear and quadratic SVM do not perform that well; for the Gaussian kernel on $\dQ$ and $\dQ^\pi$ the mean error is only $16\%$ to $20\%$, and $19\%$ to $21\%$ for $\dQk$.  The overall best is ${\dQ}_1^\pi$ achieving a mean error of $16.59\%$.

\subsection{Classifying Trajectories 3: Landmark-Sensitivity}
\label{Classifying Trajectories 3: Landmark-Sensitivity}

To show the further advantage of $\dQ$ and $\dQ^\pi$, we create a synthetic data set that appears random, except one set of trajectories pass nearby a POI and the others do not.
We randomly generate two classes of trajectories on the map of Beijing, and each class has 30 trajectories. Each trajectory has 10 critical points, and all blue trajectories passes through some point close to the city center, and all pink trajectories do not.
We hand-pick a point at the Palace Museum, the center of the city, and randomly choose other 9 points to form the set $Q$.
As shown in Figure \ref{data_synthetic and real data 2}(Right), these trajectories are a mess and largely indistinguishable, except that the blue set passes near the landmark: Palace Museum.  We next show that $\dQ$ and $\dQ^\pi$ which are landmark-aware (e.g., POI-aware) have significantly more power in distinguishing these classes.

\begin{figure}[b]
	\centering
	\includegraphics[width=3.1cm]{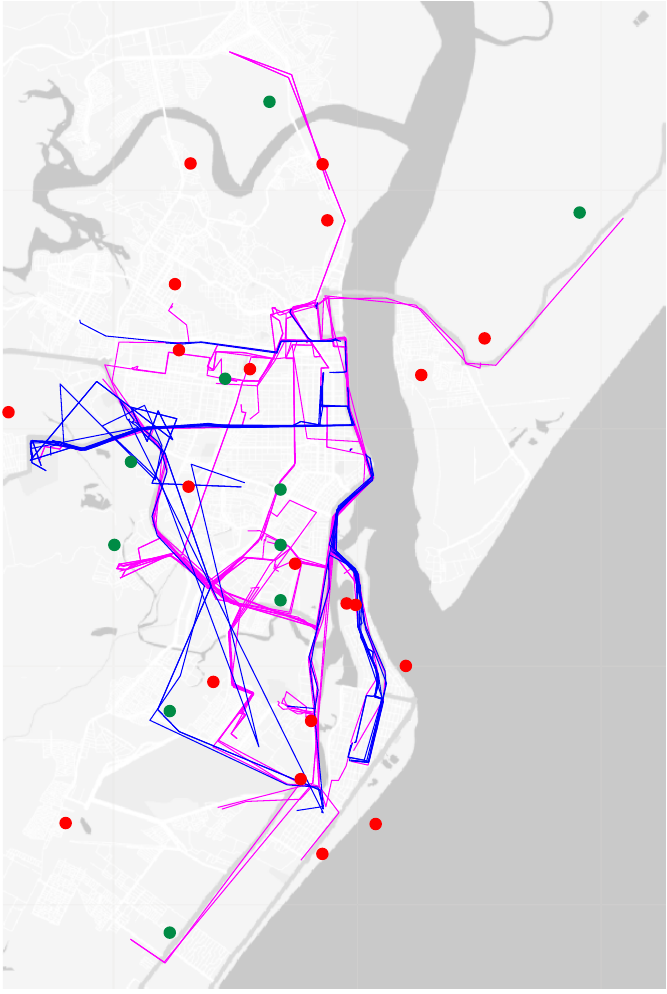}
	\hspace{0.02cm}
	\includegraphics[width=4.6cm]{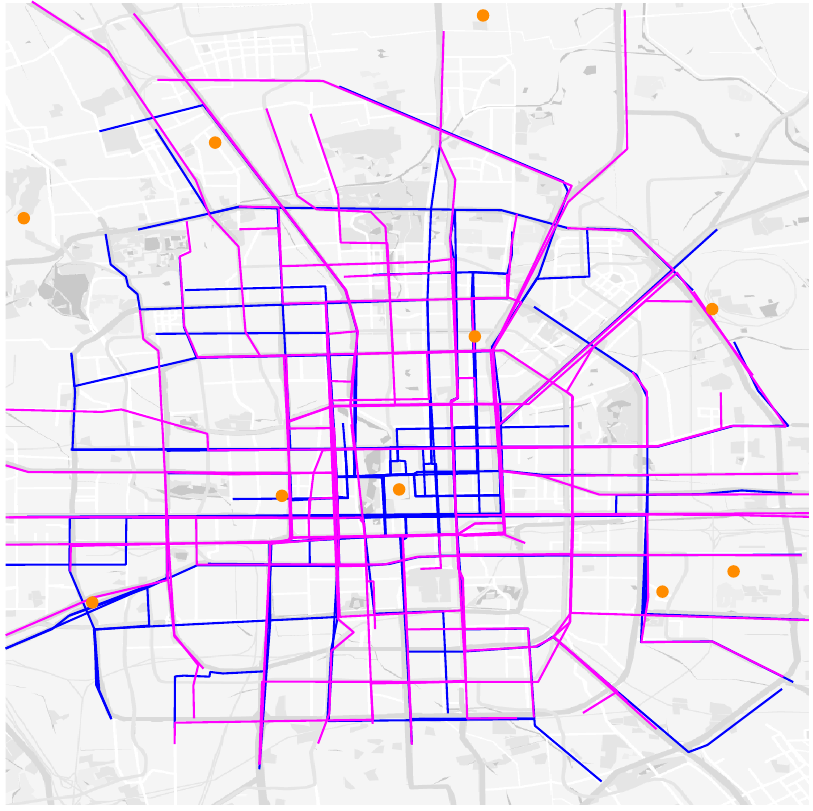}
	
	\vspace{-0.1in}
	\caption{\small Left: Bus (blue) and car (pink) trajectories with landmark sets $Q_1$ (green points), $Q_2$ (red points).
		Right: Two classes of trajectories and $Q$ (orange points).}
	\label{data_synthetic and real data 2}
	
\end{figure}

We randomly choose 21 trajectories from each class to form a training data set of size 42, and use the other trajectories as test data.  Each time, we record the error, and repeat this 1000 times to output the mean, median and variance of these errors.

\begin{table}[b]
	\vspace{-1mm}
	\center
	\caption{Landmark-sensitive classification error with KNN.}  \label{table KNN the synthetic_data}
	\vspace{-.1in}
	\begin{tabular}{r|ccc}
		\hline
		distance &  mean  &  median  & variance \\
		\hline
		Eu & 0.5226 & 0.5000 & 0.0100 \\
		dF & 0.5056 & 0.5000 & 0.0096 \\
		DTW & 0.4777 &0.5000 & 0.0107 \\
		dH & 0.4627 & 0.4444& 0.0105\\
		LCSS ($\varepsilon=0.001,\delta=8$) &0.3437  & 0.3333 & 0.0066 \\
		EDR($\varepsilon=0.02$) &  0.3916 &  0.3889  & 0.0068 \\
		LSH1$_Q$ (r=0.01) & 0.2524 & 0.2222 & 0.0098 \\
		LSH2$_Q$ (r=0.02) & 0.3248 & 0.3333 & 0.0084 \\
		\hline
		$\dQ$  & 0.4729 & 0.5000 & 0.0101 \\
		$\dQW$ ($w_1=0.3$)  & 0.4133 & 0.3889 & 0.0111 \\
		$\dQW$ ($w_1=0.6$)  & 0.2687 & 0.2778 & 0.0094\\
		$\dQW$ ($w_1=0.9$)  & 0.0592 & 0.0556 & 0.0037\\
		\hline
		$\dQ^\pi$  &  0.4385 & 0.4444 &  0.0092 \\
		$\dQW^\pi$ ($w_1=0.3$) & 0.3846 & 0.3889 & 0.0085 \\
		$\dQW^\pi$ ($w_1=0.6$) & 0.2396 & 0.2222  & 0.0065 \\
		$\dQW^\pi$ ($w_1=0.9$) & 0.1002 & 0.0556 & 0.0067 \\
		\hline
		$\dQk$  & 0.4711  & 0.4444 & 0.0106  \\
		$\dQWk$ ($w_1=0.3$) & 0.4468 & 0.4444 & 0.0113 \\
		$\dQWk$ ($w_1=0.6$) & 0.4377 & 0.4444 & 0.0112 \\
		$\dQWk$ ($w_1=0.9$) & 0.4466 & 0.4444 & 0.0100 \\
		\hline
	\end{tabular}
	\vspace{-2mm}
\end{table}

Table \ref{table KNN the synthetic_data} shows the KNN classification results.  Distances Eu and dF provide no advantage over a random classifier (which would report error $0.5$).  $\dQ^\pi$, $\dQ$, $\dQk$, DTW, and Hausdorff achieve only slight advantage over random classifiers, with error rates about $43\%$ to $48\%$, with the best achieved by $\dQ^\pi$.  This extends to the SVM approaches in Table \ref{table SVM_synthetic data}.  The best parameter free approach is $\dQ^\pi$ at $43.85\%$ error.
The parameterized distances LCSS, EDR, LSH$1_Q$, and LSH$2_Q$ perform better with error rates $25\%$ to $40\%$; but these can be sensitive to the parameter choices -- we only show the best results.

%

\begin{table}[b]
	\center
	\vspace{0mm}
	\caption{Landmark-sensitive classification error with SVM.}  \label{table SVM_synthetic data}
	\vspace{-.1in}
	\begin{tabular}{cr|cccc}
		\hline
		kernel & statistics& $\dQk$ & $\dQ $ & $\dQ^\pi$ & Eu \\
		\hline
		\multirow{3}{*}{linear}    &mean     & 0.5000 & 0.4586 & 0.4941 & 0.5887 \\
		                           &median   & 0.5000 & 0.4444 & 0.5000 & 0.6111 \\
	  	                           &variance & 0.0095 & 0.0097 & 0.0099 & 0.0085 \\
		\hline
		\multirow{3}{*}{quadratic} &mean     & 0.5403 & 0.4617 & 0.5574 & 0.4795 \\
		                           &median   & 0.5556 & 0.4444 & 0.5556 & 0.5000 \\
		                           &variance & 0.0092 & 0.0094 & 0.0101 & 0.0112 \\
		\hline
		\multirow{3}{*}{Gaussian}  &mean     & 0.5059 & 0.4567 & 0.4556 & 0.5906 \\
		                           &median   & 0.5000 & 0.4444 & 0.4444 & 0.6111 \\
		                           &variance & 0.0092 & 0.0089 & 0.0099 & 0.0088 \\
		\hline
	\end{tabular}
\end{table}

\begin{table}[b]
	\vspace{-.1in}
	\center
	\caption{Landmark-sensitive classification error with weighted Gaussian SVM.}
	\label{table SVM synthetic data gaussian}
	\vspace{-.1in}
	\begin{tabular}{c|ccc}
		\hline
		metrics &  mean  &  median  & variance \\
		\hline
		$\dQW$ ($w_1=0.3$)  &  0.1487 & 0.1667    & 0.0065  \\
		$\dQW$ ($w_1=0.6$)  &  0.0303 & 0 & 0.0014 \\
		$\dQW$ ($w_1=0.9$)  &  \pmb{0.0159} & 0 &  0.0007 \\
		\hline
		$\dQW^\pi$ ($w_1=0.3$) & 0.2997  & 0.2778 & 0.0088  \\
		$\dQW^\pi$ ($w_1=0.6$) &  0.1053 & 0.1111 & 0.0049  \\
		$\dQW^\pi$ ($w_1=0.9$) & 0.0316  & 0      & 0.0015 \\
		\hline
		$\dQWk$ ($w_1=0.3$) & 0.4942 & 0.5000 & 0.0095 \\
		$\dQWk$ ($w_1=0.6$) & 0.4726 & 0.5000 & 0.0095 \\
		$\dQWk$ ($w_1=0.9$) & 0.4687 & 0.4444 & 0.0095 \\
		\hline
	\end{tabular}
\end{table}

Next we can consider re-weighting the importance of the landmarks $Q$, for instance in the case where one particular POI (in this case $q_1$) is known to have a specific meaning in the classification task (e.g., did someone stop by the sporting event, or a military point of interest).
Suppose $w_i>0$ is a weight of $q_i\in Q$, and $W=(w_1,w_2,...,w_n)$.
Then we can generalize the definitions to:
\begin{equation*} 
\small
\begin{split}
\dQW(\gamma^{(1)},\gamma^{(2)})=&\Big(\sum\nolimits_{i=1}^{n}w_i\big(d_i^{(1)}-d_i^{(2)}\big)^2\Big)^{\frac{1}{2}},\\
\dQW^\pi(\gamma^{(1)},\gamma^{(2)})=&\sum\nolimits_{i=1}^{n}w_i\big(\big\|p_i^{(1)}-p_i^{(2)}\big\|\big).
\end{split}
\normalsize
\end{equation*}
Let $w_1\in(0,1)$ be the weight of $q_1$, and $w_i=\frac{1}{9}(1-w_1)$  (for $2\leq i \leq 10$) be the weight of all other points in $Q$.

Now observe in Table \ref{table KNN the synthetic_data} that the landmark-based distance using a KNN classifier can achieve very low error ($6\%$ for $\dQW$ and $10\%$ for $\dQW^\pi$) as we gradually increase the weight of the point $q_1$ from $w_1 = 0.1$ (i.e., $\dQ$ or $\dQ^\pi$) to $w_1 = 0.9$ to emphasize a desired POI.
The result is even more pronounced for the Gaussian SVM, as shown in Table \ref{table SVM synthetic data gaussian}; similar plots are shown for linear and quadratic kernels in Appendix \ref{SVM_linear_quadratic_synthetic_data}.  As $w_1$ is increased from (uniform) $0.1$ to $0.9$, the mean error decreases from $45\%$ to $1.5\%$ for $\dQW$ and from $45\%$ to $3\%$ for $\dQW^\pi$.
Thus, while all other distances we tried are only slightly better than random unless their parameters are tuned, by emphasizing a particular POI (a very intuitive adjustment), we achieve almost no error in classifying these trajectories.

\subsection{Using $\dQ$ in Nearest Neighbor Search}
\label{using dQ in knn}
We demonstrate that $\dQ$'s sketched representation of the trajectories in $\mathbb{R}^{|Q|}$ allows for \emph{extremely efficient k-nearest neighbor search}.
We consider two representative methods~\cite{XLP2017,SLB2018} for comparison; but do all, e.g., \cite{YRWSX2016}) which require timing information.

As a first comparison, consider a recent heavily-optimized kNN search algorithm focusing on Hausdorff and dF distances~\cite{XLP2017}; this system, DFT, is optimized for distributed algorithms on a cluster, but show results on $1$ node which we compare against.
We obtained a random sample of the GEN-TRAJ data set containing $m$ = 3 million trajectories, using 36GB of storage (larger than their 30.9GB dataset~\cite{XLP2017}).  From \emph{their} Figure 10, their indexes take 2000 to 6000 sec to build, and kNN queries require 50 to 200 seconds for $k=10$.

Another distributed system DITA~\cite{SLB2018} for trajectory similarity search focuses on DTW, returning all trajectories within a threshold.  In their~\cite{SLB2018} Figures 7(a) and 8(a), using 256 cores they achieve query time between $0.001$ and $0.01$ seconds on Beijing (10.4GB) and Chengdu (28GB) datasets.

To perform kNN queries using $\dQ$ we can sketch trajectories as $|Q|$-dimensional vectors and use Euclidean distance.  Hence, once we create the sketches, we can use any of the highly optimized packages for kNN Euclidean queries (c.f., \url{http://ann-benchmarks.com}); we choose a consistent top performer K-Graph (\url{https://github.com/aaalgo/kgraph}) with settings: recall=0.99 and max\_iteration=50.
We run on a desktop with a 6-core Intel Xeon CPU ES-1650 v3 @3.5GHz processor, and 128GB RAM; the same processor as in DFT~\cite{XLP2017}.

For experiments, we randomly choose a set of landmarks among the trajectories with $|Q| = \{12, 20, 28, 36, 44, 52\}$.
From these $Q$ we preprocess the data to derive $m \times |Q|$ sketches, a txt file we pass to K-Graph.  Then K-Graph builds an index, and allows queries.  The preprocessing time (to build sketch), sketch file size, time to build K-Graph's index, that index size, and the average query time are shown in Table \ref{table running time experiment}. For all these different values of $|Q|$, the K-Graph algorithm reaches recall=0.99 within 7 iterations.

The preprocessing and index building times take $38$ to $160$ seconds and $106$ to $129$ seconds, respectively.  By comparison, it takes $673$ seconds to load the raw data into memory.  Combined they are an order of magnitude faster than the index build time for Hausdorff in DFT~\cite{XLP2017}.  The sketch size is only $300$ to $1500$ MB, and the index sizes are $1000$ MB; reducing the size by $1$ or $2$ orders of magnitude from the original size.
Finally, the query times are only $0.00032$ to $0.00042$ seconds; that is $5$ orders of magnitude faster than the DFT index optimized for Hausdorff distance! and $1$ to $2$ orders of magnitude faster than DITA optimized for DTW and using 256 cores on smaller data.
Thus, using $\dQ$ (and existing libraries) allows for small data sketches, and extremely efficient kNN queries.

\begin{table}[t]
	\center
	\caption{The running time experiment of KNN search.}
	\label{table running time experiment}
	\vspace{-.1in}
	\begin{tabular}{r|C{0.4cm}C{0.4cm}C{0.4cm}ccc}
		\hline
		$|Q|$  & 12 & 20  & 28 & 36 & 44 & 52  \\
		\hline
		preprocessing time (s) &38  &62  & 88 &114  & 138& 160\\
		sketch size (MB) &337  & 560 &785  &1012  & 1331 & 1536 \\
		index time (s) & 106 &109 &114 & 119 &124 &129 \\
		index file size (MB) & 999 & 999 & 1005 & 1002 & 1007 & 1001 \\
		query time ($10^{-4}$s)& 4.2 & 3.7 & 4.2 & 3.2& 3.5 & 3.7 \\
		\hline
	\end{tabular}
\vspace{-4mm}
\end{table}

\vspace{-2mm}
\section{Conclusion and Discussion}
\label{sec:conclude}
\vspace{-2mm}
We introduce a new family of landmark-based distances $\dQ$, with applications to trajectories and hyperplanes (regressors, separators).  These have nice mathematical properties, e.g., being psuedo-metrics, metrics, and bounded VC-dimension metric balls.  On trajectories, new metrics $\dQ$ and $\dQ^\pi$ are the most general and best or competitive against all other distances in \emph{all} analysis tasks; see Table \ref{tbl:all}.

\begin{table}[b!]
\caption{\label{tbl:all}
Distances on analysis tasks as: best $\bullet$, competitive {\color{gray} $\bullet$}, near competitive $\circ$; possible \checkmark or possible but slower {\tiny \checkmark}.}

\vspace{-2mm}
\begin{tabular}{r|ccccccccccc}
\hline
 task & \dQ & $\dQ^\pi$ & \dQk & \hspace{-1mm}Eu & \hspace{-1mm}dF & \hspace{-1mm}DTW\hspace{-1mm} & \hspace{-1mm}dH & \hspace{-2mm}LCSS\hspace{-2mm} & EDR &
 \hspace{-2mm}LSH$_Q$\hspace{-2mm}
\\  \hline
\hspace{-2mm}easy clust & \checkmark & \checkmark & \checkmark & \checkmark & - & - & - & - & - & -
\\
learn 1 & {\color{gray} $\bullet$} & $\bullet$  & -  & {\color{gray} $\bullet$} &  $\circ$  & {\color{gray} $\bullet$} & {\color{gray} $\bullet$} & {\color{gray} $\bullet$} & {\color{gray} $\bullet$} & -
\\
\hspace{-2mm}learn 2 & {\color{gray} $\bullet$} & $\bullet$  & $\circ$  & $\circ$ & -  & {\color{gray} $\bullet$} & - & $\circ$ & {\color{gray} $\bullet$} & $\circ$
\\
\hspace{-2mm}learn 3 & $\bullet$ & {\color{gray} $\bullet$}  & -  & - & -  & - & - & - & - & -
\\
fast NN & \checkmark & {\tiny \checkmark} & {\tiny \checkmark}  & {\tiny \checkmark} & - & {\tiny \checkmark} & - & - & - & {\tiny \checkmark}
\\ \hline
 any $k$ & \checkmark & \checkmark & - & -  & \checkmark & \checkmark & \checkmark & \checkmark & \checkmark & \checkmark
\\ \hline
\end{tabular}
\end{table}

The landmark set $Q$ can be randomly chosen and small, or its points can hold specific meaning in which case, the interpretation and discriminatory ability of the distances are greatly enhanced.
A companion paper~\cite{PT19} provides an in depth theoretical study of how many landmarks are required to preserve certain errors, how to chose them, and when curves can be explicitly recovered from them.  In the present paper, we simply empirically show that in most cases $20$ random landmarks are sufficient.

These provide meaningful \emph{vectorized representations}.
They are general and simple to compute and work with.  We believe many applications of these sorts of vectorized distances will be discovered.  And there are more mathematical questions to ask about the geometric and statistical power of these landmark-based distances.

\paragraph*{Software.}
Code for reproducing experiments
in Section \ref{sec:traj-analysis}
is available here
\href{https://drive.google.com/open?id=1Z_Na1nfioM_We8b1FnTU5UVuOYCjbP-j}
{https://drive.google.com/open?id=1Z\_Na1nfioM\_We8b1F}
\href{https://drive.google.com/open?id=1Z_Na1nfioM_We8b1FnTU5UVuOYCjbP-j}{nTU5UVuOYCjbP-j}

\vspace{-3mm}
\bibliographystyle{ACM-Reference-Format}


\paragraph*{Appendix.}
A full version with appendix for extended proofs and more experimental details available here~\cite{PT19a}.

\clearpage

\appendix

\section{Metric Property for Unsigned Variant on the Distance}
\label{sec:metric}

Suppose $Q=\{q_1,q_2,\cdots,q_n\}\subset \R^2$, $\ell_1,\ell_2\in \c{L}=\{\ell \mid  \ell \text{ is a line in }$ $\R^2 \}$.
%
Given $\ell\in \c{L}$, we write $\ell$ in the form as before and define
$\bar{v}_Q(\ell)=\left(\bar{v}_{Q_1}(\ell),\bar{v}_{Q_2}(\ell),\ldots,\bar{v}_{Q_n}(\ell)\right)$ where
$\bar{v}_{Q_i}(\ell)= |u_1 x_i+ u_2 y_i + u_3|$ and $(x_i,y_i)$ is the coordinates of $q_i\in Q$,
and then define the first variant of $\dQ$ as
\begin{small}
	\begin{equation} \label{first variant of dQ}
	\begin{split}
	&\ddQ(\ell_1,\ell_2):=\big\|\frac{1}{\sqrt{n}}(\bar{v}_Q(\ell_1)-\bar{v}_Q(\ell_2))\big\|\\
	=&\big(\sum_{i=1}^n\frac{1}{n}(\bar{v}_{Q_i}(\ell_1)-\bar{v}_{Q_i}(\ell_2))^2\big)^{\frac{1}{2}}.
	\end{split}
	\end{equation}
\end{small}
For \eqref{first variant of dQ}, we have the following theorem.

\begin{theorem}  \label{metric property variant1 of dQ}
	Suppose in $Q\subset \R^2$ there is a subset of five points, and any three points in this subset are non-collinear, then  $\ddQ$ is a metric in $\c{L}$.
\end{theorem}

\begin{proof}
	We only need to show if $\ddQ(\ell_1,\ell_2)=0$, then $\ell_1=\ell_2$.
	Suppose $\widetilde{Q}=\{q_1,\cdots, q_5\}\subset Q$, and any three points in $\widetilde{Q}$ are not on the same line.
	If $\ell_1 \neq \ell_2$, then let $\ell_1'$ and $\ell_2'$ be the two bisectors of the angles formed by $\ell_1$ and $\ell_2$. From $\ddQ(\ell_1,\ell_2)=0$, we know $\bar{v}_{Q_i}(\ell_1)=\bar{v}_{Q_i}(\ell_2)$ for $i\in[5]$, which means the distances from $q_i\in \widetilde{Q}$ to $\ell_1$ and to $\ell_2$ are equal. So, any point $q_i\in \widetilde{Q}$ must be either on  $\ell_1'$ or on $\ell_2'$, which implies there must be three collinear points in $\widetilde{Q}$.
	This is contradictory to the fact that any three points in $\widetilde{Q}$ are not on the same line.
\end{proof}

\paragraph*{Remark.}
Definition \eqref{first variant of dQ} can be generalized to hyperplanes in $\mathbb{R}^d$:
\begin{small}
	\begin{equation} \label{first variant of dQ in R^d}
	\ddQ(h_1,h_2):=\big(\sum_{i=1}^n\frac{1}{n}(\bar{v}_{Q_i}(h_1)-\bar{v}_{Q_i}(h_2))^2\big)^{\frac{1}{2}},
	\end{equation}
\end{small}
where $h_1,h_2 \in \c{H}=\{h\ |\  h \text{ is a hyperplane  in } \R^d \}$, and $\bar{v}_{Q_i}(h_j)$ is the distance from point $q_i$ in $Q\subset \mathbb{R}^d$ to $h_j$ ($j=1,2$). Using the similar method, we can show if there is a subset of $2d+1$ points
in $Q$ and any $d+1$ points in this subset are not on the same hyperplane, then \eqref{first variant of dQ in R^d} is a metric in $\c{H}$.

\subsection{Matrix Norm Variant}
In another variant of $\dQ$ we define $\tilde{v}_{Q_i}(\ell)$ as a \emph{vector} from $q_i$ to the closest point on $\ell$.
More specifically, suppose $\ell$ is in the same form as before, then the projection of point $q_i=(x_i,y_i)$ on $\ell$ is $(\tilde{x}_i,\tilde{y}_i)=(x_i\cos^2(\alpha)-y_i \sin(\alpha)\cos(\alpha)-c\sin(\alpha),
-x_i\cos(\alpha)\sin(\alpha)+y_i \sin^2(\alpha)-c\cos(\alpha))$, and we define $\tilde{v}_{Q_i}(\ell)=(\tilde{x}_i-x_i,\tilde{y}_i-y_i)$ for $(x_i,y_i)\in Q$, and an $n \times 2$ matrix $V_{Q,l}=[\tilde{v}_{Q_1}(\ell);\cdots;\tilde{v}_{Q_n}(\ell)]$ where $\tilde{v}_{Q_i}(\ell)$ is the $i$th row of $V_{Q,l}$. For $\ell_1,\ell_2\in \c{L}$ we define the distance between these two lines as
\begin{small}
	\begin{equation} \label{second variant of dQ}
	\DQ(\ell_1,\ell_2):=\|V_{Q, \ell_1} - V_{Q, \ell_2}\|_F,
	\end{equation}
\end{small}
where $\|\cdot\|_F$ is the Frobenius norm of matrices.
For \eqref{second variant of dQ}, we have the following theorem.

\begin{theorem}  \label{metric property variant2 of dQ}
	Suppose in $Q\subset \R^2$ there are two different points $q_1$ and $q_2$, then $\DQ$ is a metric in $\c{L}$.
\end{theorem}

\begin{proof}
	We only need to show if $\DQ(\ell_1,\ell_2)=0$, then $\ell_1=\ell_2$.
	There are two cases.
	
	(1) $\tilde{v}_{Q_1}(\ell_1)=(0,0)$ and $\tilde{v}_{Q_2}(\ell_1)=(0,0)$.
	From $\DQ(\ell_1,\ell_2)=0$ we know $\tilde{v}_{Q_1}(\ell_2)=0$ and $\tilde{v}_{Q_2}(\ell_2)=0$, which means
	$q_1$ and $q_2$ are on both $\ell_1$ and $\ell_2$, so $\ell_1=\ell_2$.
	
	(2) $\tilde{v}_{Q_1}(\ell_1)\neq(0,0)$  or $\tilde{v}_{Q_2}(\ell_1)\neq (0,0)$. In this case, without loss of generality
	we assume  $\tilde{v}_{Q_1}(\ell_1)\neq (0,0)$. From $\DQ(\ell_1,\ell_2)=0$ we have $\tilde{v}_{Q_1}(\ell_2)=\tilde{v}_{Q_1}(\ell_1)\neq (0,0)$, so introducing the notation $(\tilde{x}_i-x_i,\tilde{y}_i-y_i)=\tilde{v}_{Q_1}(\ell_1)$, we know $(\tilde{x}_i, \tilde{y}_i)$ is on
	$\ell_1$ and $\ell_2$, and $\tilde{v}_{Q_1}(\ell_1)$ is the normal direction of $\ell_1$ and $\ell_2$. Since a point and a normal direction can uniquely determine a line, we have $\ell_1=\ell_2$.
\end{proof}

\paragraph*{Remark.}
Definition \eqref{second variant of dQ} can be generalized to hyperplanes in $\mathbb{R}^d$:
\begin{small}
	\begin{equation} \label{second variant of dQ in R^d}
	\DQ(h_1,h_2):=\|V_{Q, h_1} - V_{Q, h_2}\|_F,
	\end{equation}
\end{small}
where $h_1,h_2 \in \c{H}=\{h\ |\  h \text{ is a hyperplane  in } \R^d \}$,
and $V_{Q, h_j}$ ($j=1,2$) is an $n\times d$ matrix with each row being a projection vector from a point in $Q$ to $h_j$. Using the similar method, we can show if there are $d$ different points in $Q$, then \eqref{second variant of dQ in R^d} is a metric in $\c{H}$.

\section{Metric Properties for $\dQ$ on Trajectories}
\label{metric peroperties for dQ1 and dQ2}

In this section, we prove Theorem \ref{well-difined metric for dQ1 and dQ2} for $\dQ$. We first introduce the following lemma.

\begin{lemma}  \label{lemma metric property dQ}
	As shown in Figure \ref{fig lemma metric property dQ}, suppose the line $\ell_2$ passes through $q_1$ and $c$, $\ell_1$ is perpendicular to $\ell_2$ at $q_1$, and $c$ is on the right side of $\ell_1$.	If $q_2$ is outside the circle $C(q_1,\|q_1-c\|)$, on the left side of $\ell_2$ and above $\ell_2$ (the yellow-shaded region)), and $q_3$ is outside the circle $C(q_1,\|q_1-c\|)$, on the left side of $\ell_2$ and below $\ell_2$ (the pink-shaded region), then we have $B(q_1,\|q_1-c\|)\subset B(q_2,\|q_2-c\|)\cup B(q_2,\|q_2-c\|)$.
\end{lemma}

\begin{figure}[t] 
	\includegraphics[width=0.43\linewidth]{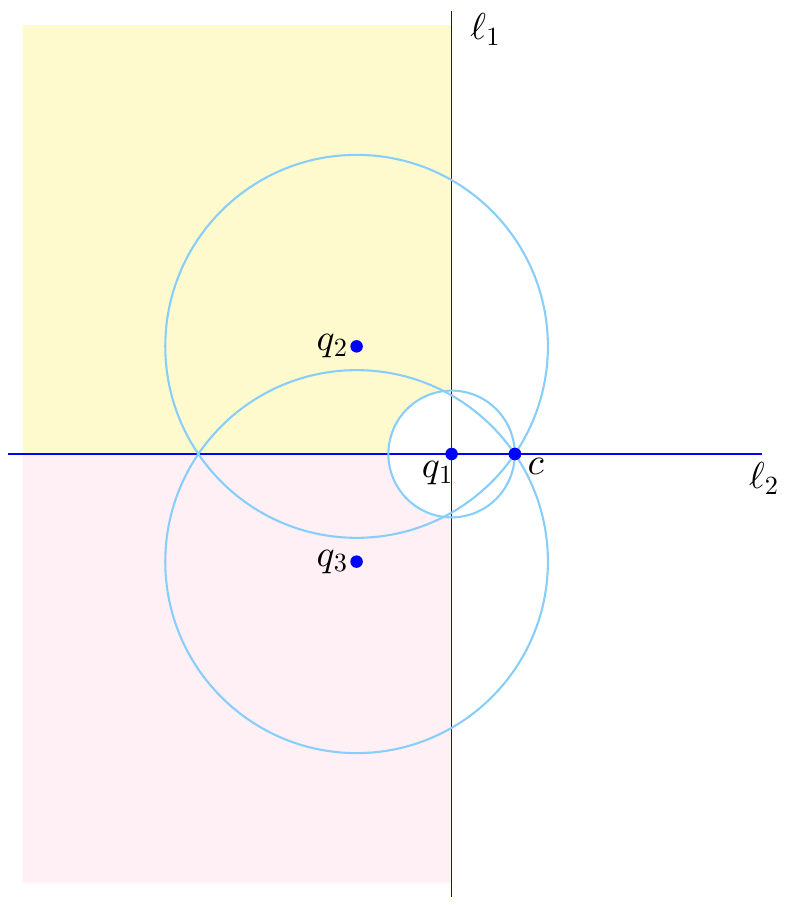}
	\hspace{0.02cm}
    \includegraphics[width=0.49\linewidth]{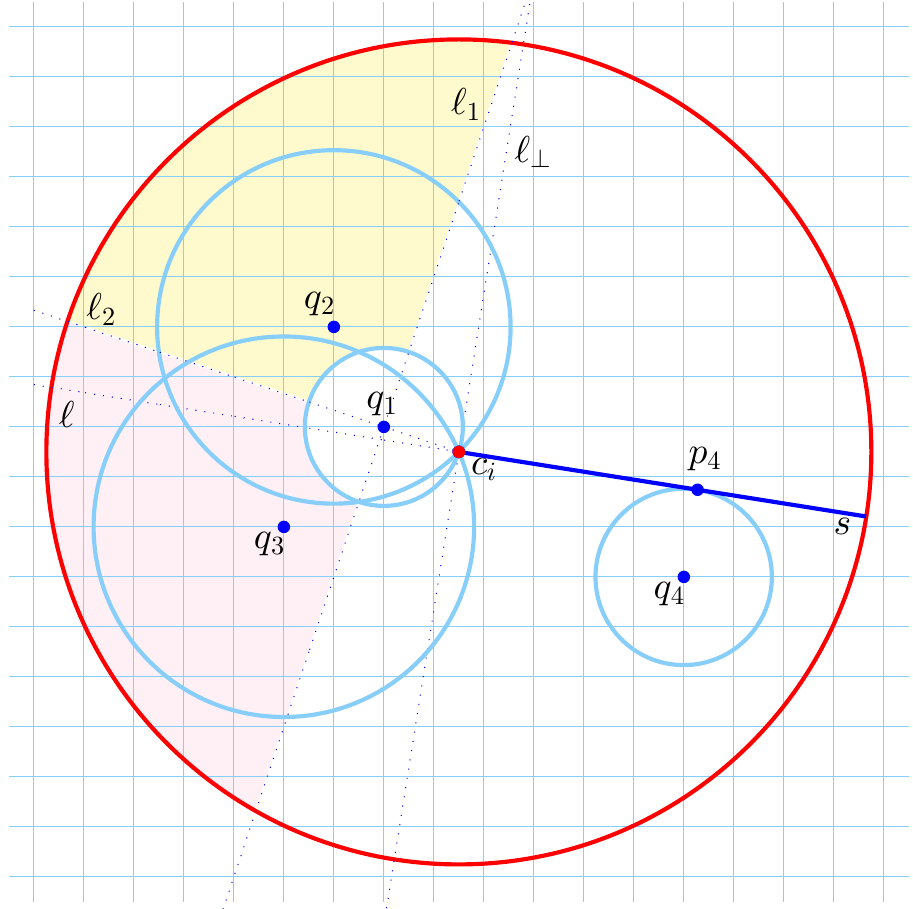}
	\caption{ Left: $\ell_1\perp \ell_2$ and $B(q_1,\|q_1-c\|)\subset B(q_2,\|q_2-c\|) \cup B(q_3,\|q_3-c\|)$.
              Right: $c_i$ is a critical point of $\gamma^{(1)}$ and $B(q_1,\|q_1-c_i\|)\subset B(q_2,\|q_2-c_i\|) \cup B(q_3,\|q_3-c_i\|)$.  } \label{fig lemma metric property dQ}
\end{figure}

\begin{figure}[t] 
	\includegraphics[width=0.49\linewidth]{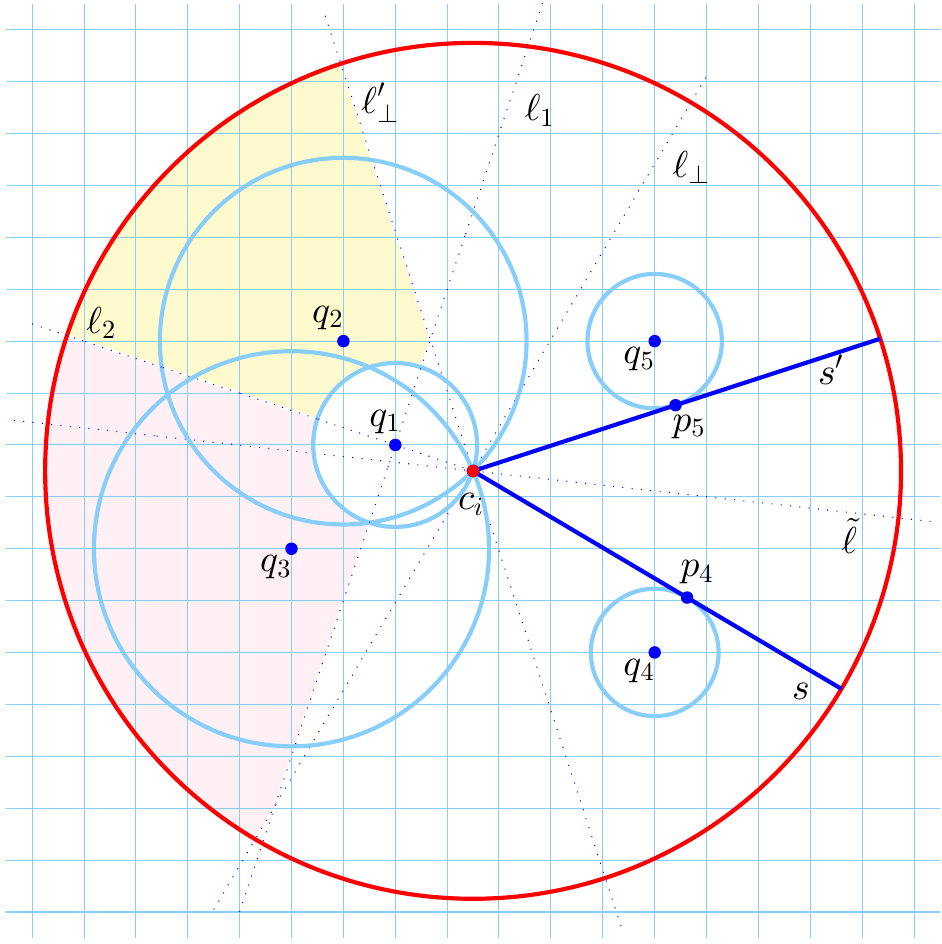}
	\hspace{0.02cm}
    \includegraphics[width=0.49\linewidth]{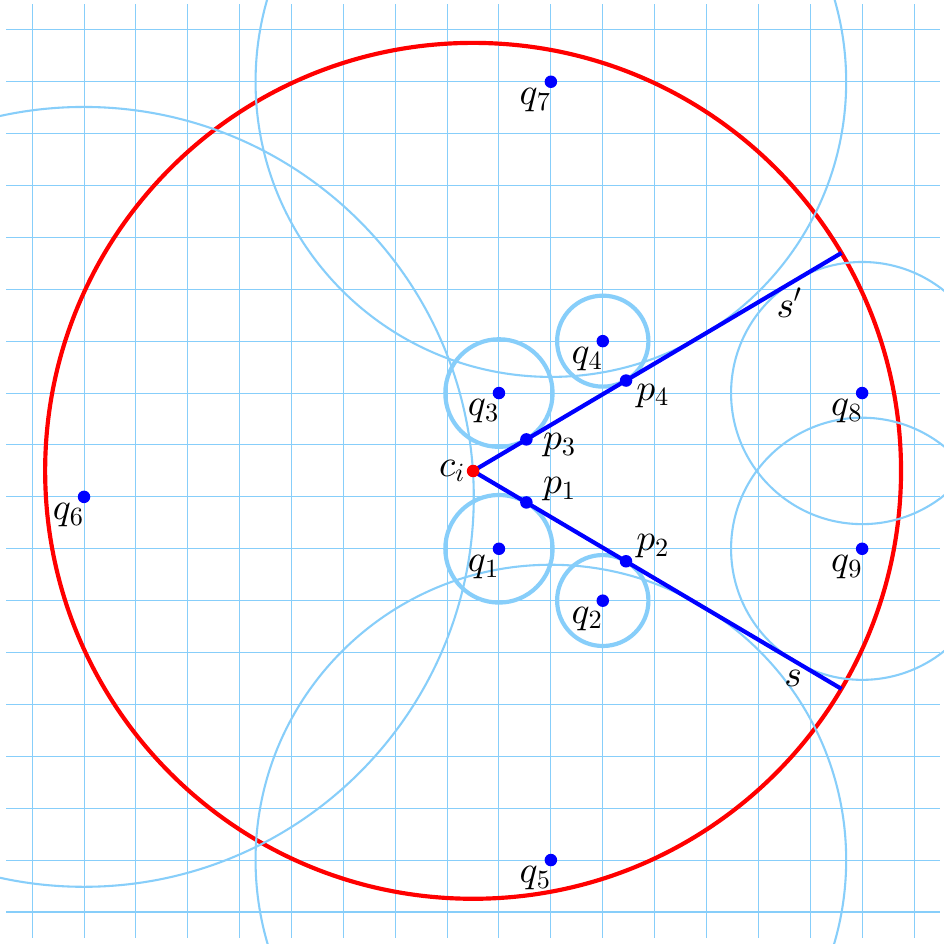}
	\caption{Left: $c_i$ is a critical point of $\gamma^{(1)}$ and $B(q_1,\|q_1-c_i\|)\subset$ $B(q_2,\|q_2-c_i\|)\cup$$B(q_3,\|q_3-c_i\|)$.
             Right: $B(q_1,\|q_1-p_1\|)$, $B(q_2,\|q_2-p_2\|)$ are tangent to $s$, and $B(q_3,\|q_3-p_3\|)$, $B(q_4,\|q_4-p_4\|)$ are tangent to $s'$. For each one of these four circles, any tangent line segment, except $s$, $s'$ cannot be extended outside $B(c_i,\frac{\tau}{2})$ without intersecting with any other circle.} \label{fig metric property dQ}
\end{figure}

\begin{proof}
	We use $q_1$ as the origin, $\ell_2$ as the $x$-axis and $\ell_1$ as the $y$-axis to build a coordinate system, and assume the coordinates of $c$, $q_2$ and $q_3$ are $(r,0)$, $(x_2,y_2)$ and $(x_3,y_3)$ respectively.
	So, we have $x_2^2+y_2^2> r^2$, $x_2^2+y_2^2> r^2$ and $x_2,x_3<0$, $y_2>0$ and $y_3<0$. Our goal is to prove if $x^2+y^2\leq r^2$ then either
	\begin{equation} \label{the case y>0}
	(x-x_2)^2+(y-y_2)^2\leq (x_2-r)^2+y_2^2,
	\end{equation}
	or
	\begin{equation}\label{the case y<0}
	(x-x_3)^2+(y-y_3)^2\leq (x_3-r)^2+y_3^2.
	\end{equation}
	
	If $y\geq0$, then from $x\leq r, x_2<0, y_2>0$ we have $(r-x)x_2\leq y y_2$, which is equivalent to $-2xx_2-2yy_2\leq -2rx_2$.
	Since  $x^2+y^2\leq r^2$, we obtain $x^2-2xx_2+y^2-2yy_2\leq -2rx_2+r^2$, which implies \eqref{the case y>0} is true.
	Similarly, if $y\leq 0$ then we can show \eqref{the case y<0} is true. Thus, the proof is completed.
	
\end{proof}

Now, we can give the proof of Theorem \ref{well-difined metric for dQ1 and dQ2} for $\dQ$.

\begin{proof}
	
	Suppose $\dQ(\gamma^{(1)},\gamma^{(2)})=0$, we only need to prove $\gamma^{(1)}=\gamma^{(2)}$.
    We draw a ball $B(c_i,\frac{1}{2}\tau)$ at a critical point $c_i$ ($1\leq i \leq k-1$) of $\gamma^{(1)}$.
    There are three possibilities.

    Case 1. As shown in Figure \ref{fig lemma metric property dQ}(Right), $c_i$ is an endpoint of $\gamma^{(1)}$, and $B(c_i,\frac{\tau}{2})$ contains one line segment $s$ of $\gamma^{(1)}$. In this case, we assume $s$ is part of line $\ell$, and draw a line $\ell_\perp$  through $c_i$ which is perpendicular to $\ell$. Then, we choose a point $q_1$ from $Q\cap B(c_i,\frac{\tau}{2})$, which is on the left side of $\ell_\perp$, close to $\ell$ and satisfies $\|q_1-c\|<2\eta$. Suppose $\ell_2$ is the line through $q_1$ and $c_i$, and $\ell_1$ is perpendicular to $\ell_2$ at $q_1$. We choose a point $q_2\in Q\cap B(c_i,\frac{\tau}{2})$ from the region that is outside $B(q_1,\|q_1-c_i\|)$, on the left side of $\ell_1$ and $\ell_\perp$, and above $\ell_2$ (the yellow-shaded region), and choose a point $q_3\in Q\cap B(c_i,\frac{\tau}{2})$ from the region that is outside $B(q_1,\|q_1-c_i\|)$, on the left side of $\ell_1$ and $\ell_\perp$, and below $\ell_2$ (the pink-shaded region). Obviously, $\{c_i\}=C(q_1,\|q_1-c_i\|) \cap C(q_2,\|q_2-c_i\|) \cap C(q_3,\|q_3-c_i\|)$, and from Lemma \ref{lemma metric property dQ}, we know $B(q_1,\|q_1-c_i\|) \subset B(q_2,\|q_2-c_i\|) \cup B(q_3,\|q_3-c_i\|)$. So, $c_i$ must be on $\gamma^{(2)}$. Since the tangent line of $C(q_1,\|q_1-c_i\|)$ at $c_i$ goes into the interior of $B(q_2,\|q_2-c_i\|)$ and $B(q_3,\|q_3-c_i\|)$, from $(O1)$ we know $c_i$
    must be a critical point of $\gamma^{(2)}$. There also exists $q_4\in B(c_i,\frac{\tau}{2})$ and $p_4\in s$ such that $B(q_4,\|q_4-p_4\|)$ is tangent to $s$ at point $p_4$. From (O1) and (O2) we know the tangent line segment of $C(q_4,\|q_4-p_4\|)$ through $c_i$ must be a part of $\gamma^{(2)}$, and this tangent line segment must be $s$ because the other tangent line segment through $c_i$ intersects with other circles. Thus, $s$ is a part of $\gamma^{(2)}$.

    Case 2. As shown in Figure \ref{fig metric property dQ}(Left), $c_i$ is an internal of $\gamma^{(1)}$, $B(c_i,\frac{\tau}{2})$ contains two line segments $s$, $s'$ of $\gamma^{(1)}$, and the angle between  $s$, $s'$ is at most $\frac{\pi}{4}$.
    In this case, we assume $\tilde{\ell}$ is the line bisecting the angle formed by $s$ and $s'$, and draw two lines $\ell_\perp$ and $\ell'_\perp$ which is perpendicular to $s$ and $s'$ at $c_i$ respectively. Then, we choose a point $q_1$ from $Q\cap B(c_i,\frac{\tau}{2})$, which is on the left side of $\ell_\perp$ and $\ell'_\perp$, close to $\tilde{\ell}$ and satisfies $\|q_1-c\|<2\eta$. Suppose $\ell_2$ is the line through $q_1$ and $c_i$, and $\ell_1$ is perpendicular to $\ell_2$ at $q_1$.  We choose a point $q_2\in Q\cap B(c_i,\frac{\tau}{2})$ from the region that is outside $B(q_1,\|q_1-c_i\|)$, on the left side of $\ell_1$, $\ell_\perp$ and $\ell'_\perp$ and above $\ell_2$ (the yellow-shaded region), and choose a point $q_3\in Q\cap B(c_i,\frac{\tau}{2})$ from the region that is outside $B(q_1,\|q_1-c_i\|)$, on the left side of $\ell_1$, $\ell_\perp$ and $\ell'_\perp$ and below $\ell_2$ (the pink-shaded region). Obviously, $\{c_i\}=C(q_1,\|q_1-c_i\|) \cap C(q_2,\|q_2-c_i\|) \cap C(q_3,\|q_3-c_i\|)$, and from Lemma \ref{lemma metric property dQ}, we know $B(q_1,\|q_1-c_i\|) \subset B(q_2,\|q_2-c_i\|) \cup B(q_3,\|q_3-c_i\|)$. So, $c_i$ must be on $\gamma^{(2)}$. Since the tangent line of $C(q_1,\|q_1-c_i\|)$ at $c_i$ goes into the interior of $B(q_2,\|q_2-c_i\|)$ and $B(q_3,\|q_3-c_i\|)$, from $(O1)$ we know $c_i$ must be a critical point of $\gamma^{(2)}$. There also exists $q_4,q_5\in B(c_i,\frac{\tau}{2})$ and $p_4\in s$, $p_5\in s'$ such that $B(q_4,\|q_4-p_4\|)$ is tangent to $s$ at point $p_4$, and $B(q_5,\|q_5-p_5\|)$ is tangent to $s'$ at point $p_5$.
    Using the similar argument in Case 1,we can show $s$ and $s'$ both belong to $\gamma^{(2)}$.

    Case 3. As shown in Figure \ref{fig metric property dQ}(Right), $c_i$ is an internal of $\gamma^{(1)}$, $B(c_i,\frac{\tau}{2})$ contains two line segments $s$, $s'$ of $\gamma^{(1)}$, and the angle between  $s$, $s'$ is greater than $\frac{\pi}{4}$. In this case, we choose four points $q_1,q_2,q_3,q_4$ from $Q\cap B(c_i,\frac{\tau}{2})$ such that the circles with center $q_1,q_2$ are tangent to $s$  at $p_1,p_2$, and the circles with center $q_3,q_4$ are tangent to $s'$ at $p_3,p_4$. Moreover, we can require $\|q_j-c_j\|\leq \eta$ for $1\leq j \leq 4 $ and these four circles do not intersect with each other.  Then, we can choose three points $q_5,q_6,q_7$ outside the angle region formed by $s$ and $s'$, and two points $q_8,q_9$ inside this angle region. Using $C_{j'}$ ($5\leq j'\leq 9$) to represent the circles corresponding to these five points, we can choose these points close to the boundary of $B(c_i,\frac{\tau}{2})$, and require $C_6$ contains $c_i$, $C_5,C_9$ are tangent to $s$, $C_7,C_8$ are tangent to $s'$, and $C_5\cap C_6\neq \emptyset$, $C_6\cap C_7\neq \emptyset$, and $C_8\cap C_9\neq \emptyset$. Thus, any tangent line segment of $C(q_j,\|q_j-p_j\|)$ ($1\leq j \leq 4$), except $s,s'$, can not be extended outside $B(c_i,\frac{\tau}{2})$ without intersecting with $\cup_{5\leq j'\leq 9}C_{j'}$. From (O1) and (O2) we know $\gamma^{(2)}$ must be tangent to $C(q_1,\|q_1-p_1\|)$ or
    $C(q_2,\|q_2-p_2\|)$, and without loss of generality we assume a tangent line segment of $C(q_1,\|q_1-p_1\|)$ is a part of $\gamma^{(2)}$.
    Since (O2), (O3) imply this tangent line segment must be extended outside $B(c_i,\frac{\tau}{2})$ without going into the interior of any other circle, we know
    $s\cap B(q_1,\delta)$ is a part of $\gamma^{(2)}$ for some $\delta>0$. Similarly, we have $s\cap B(q_3,\delta)$ is a part of $\gamma^{(2)}$ for some $\delta>0$. Since there is at most one critical point of $\gamma^{(2)}$ in $B(c_i,\frac{\tau}{2})$, from (O3) we know $c_i$ must be a critical point of $\gamma^{(2)}$. Thus, $s$ and $s'$ both belong to $\gamma^{(2)}$.

	From the discussion of above three cases, we know $\gamma^{(2)}$ overlaps with $\gamma^{(1)}$ in the ball $B(c_i,\frac{\tau}{2})$,
    and a similar argument leads to $\gamma^{(1)}=\gamma^{(2)}$.
	
\end{proof}

\section{Common Distance Measurements for Trajectories}
\label{def of distances between traojectories}

In this section, we briefly introduce the definition of Euclidian distance, discrete Frechet distance and dynamic time warping distance. Suppose $\gamma^{(1)}$ and $\gamma^{(2)}$ are two trajectories in $\R^2$ with critical points $c_0^{(1)}, c_1^{(1)},...c_{k_1}^{(1)}$ and $c_0^{(2)}, c_1^{(2)},...c_{k_2}^{(2)}$ respectively.

\paragraph*{Euclidean Distance.} It requires $k_1=k_2$ and takes the average Euclidean distance between corresponding critical points.
\begin{equation*}
\text{Eu}(\gamma^{(1)},\gamma^{(2)})=\frac{1}{k_1}\sum\nolimits_{i=0}^{k_1}\big\|c_i^{(1)}-c_i^{(2)}\big\|.
\end{equation*}

\paragraph*{Discrete Frechet Distance.} It measures the similarity between two piecewise-linear curves by taking into account their location and time ordering. Here, we introduce its definition in ~\cite{TEHM1994}.
Suppose $\mathcal{A}=\{a_0,a_1,\cdots,a_m\}\subset\{0,1,\cdots,k_1\}$,  $\mathcal{B}=\{b_0,b_1,\dots,b_m\}\subset\{0,1,\cdots,k_2\}$, and $a_0=b_0=0$, $a_m=k_1$, $b_m=k_2$. If for each $i\in\{0,\cdots,k_1-1\}$
we have  $a_{i+1}=a_i$ or $a_{i+1}=a_i+1$, and for each $i\in\{0,\cdots,k_2-1\}$, we have $b_{i+1}=b_i$ or $b_{i+1}=b_i+1$, then we say
$\mathcal{A}$ and $\mathcal{B}$ can determine a coupling $\mathcal{L}$ between $\gamma^{(1)}$ and $\gamma^{(2)}$, which is a sequence
{\small $\big(c_{a_0}^{(1)}, c_{b_0}^{(2)}\big),\big(c_{a_1}^{(1)}, c_{b_1}^{(2)}\big),\cdots,\big(c_{a_m}^{(1)}, c_{b_m}^{(2)}\big)$}.
We define the \textit{length} of $\mathcal{L}$ as $\|\mathcal{L}\|=\max_{0\leq i\leq m}\big\|c_{a_i}^{(1)}-c_{b_i}^{(2)}\big\|$. The discrete Frechet
distance is defined as:
\begin{small}
	\begin{equation*}
	\begin{split}
	&\text{dF}(\gamma^{(1)},\gamma^{(2)})\\
	=&\min\{\|\mathcal{L}\||\  \mathcal{L} \text{ is a a coupling between } \gamma^{(1)} \text{ and }\gamma^{(2)}\}.
	\end{split}
	\end{equation*}
\end{small}

\paragraph*{Dynamic Time Warping (DTW) Distance.} DTW ~\cite{YJF1998} is an algorithm to find the optimal matching between the critical points of two trajectories, and it does not require $k_1=k_2$. It is defined and computed by the recursion formula:
$D(i,j)=\big\|c_i^{(1)}-c_j^{(2)}\big\|+\min\big(D(i-1,j),\ D(i-1,j-1),\ D(i,j-1)\big)$,
where $D(0,j)=\|c_0^{(1)}-c_j^{(2)}\|$, $D(i,0)=\|c_i^{(1)}-c_0^{(2)}\|$, and DTW distance between $\gamma^{(1)}$ and $\gamma^{(2)}$ is defined as $\text{DTW}(\gamma^{(1)},\gamma^{(2)})=D(k_1,k_2)$.

\paragraph*{Discrete Hausdorff Distance.} It measure the spatial similarity between two trajectories ~\cite{FM2008}:
\begin{small}
	\begin{equation*}
	\text{dH}(\gamma^{(1)}, \gamma^{(2)})=\max (d(\gamma^{(1)}, \gamma^{(2)}),d(\gamma^{(2)}, \gamma^{(1)}))
	\end{equation*}
\end{small}
where $d(\gamma^{(1)}, \gamma^{(2)})=\max_{0\leq i\leq k_1}\min_{0\leq j\leq k_2}\|c_i^{(1)}-c_j^{(2)}\|$.

\paragraph*{Longest Common Subsequence Distance.} It finds the alignment between two sequences that maximize the length of common subsequence. Let $\text{Head}(\gamma^{(1)})$ be the first $k_1-1$ critical points of $\gamma^{(1)}$,
and $\text{Head}(\gamma^{(2)})$ be the first $k_2-1$ critical points of $\gamma^{(2)}$. Given  $\varepsilon, \delta>0$, the $\text{lcss}_{\varepsilon,\delta}(\gamma^{(1)},\gamma^{(2)})$ is defined as follows  ~\cite{ZZKH2006}:

\begin{small}
	\begin{equation*}
	\begin{split}
	&\text{lcss}_{\varepsilon,\delta}(\gamma^{(1)},\gamma^{(2)})=\\
	&\begin{cases}
	0,   \text{\ \ \ \ if $\gamma^{(1)}$ or $\gamma^{(2)}$ is empty}\\
	1+\text{lcss}_{\varepsilon,\delta}(\gamma^{(1)},\gamma^{(2)}),  \text{\ \ if $\|c_{k_1}^{(1)}-c_{k_2}^{(2)}\|<\varepsilon$ and $|k_1-k_2|<\delta$ }\\
	\max\big(\text{lcss}_{\varepsilon,\delta}(\text{Head}(\gamma^{(1)}),\gamma^{(2)}),
	\text{lcss}_{\varepsilon,\delta}(\gamma^{(1)},\text{Head}(\gamma^{(2)}))\big),\text{\ \ otherwise} \\
	\end{cases} \hspace{-0.1in}.
	\end{split}
	\end{equation*}
\end{small}
LCSS distance is defined as
$\text{LCSS}_{\varepsilon,\delta}(\gamma^{(1)},\gamma^{(2)})
=1-\frac{\text{lcss}_{\varepsilon,\delta}(\gamma^{(1)},\gamma^{(2)})}{\max (k_1,k_2)}$.

\paragraph*{Edit Distance for Real Sequences.}
It is similar to the edit distance on strings,
and seeking the minimum number of edit operations required to change one trajectory to another  ~\cite{COO2005}.
For EDR with $\varepsilon>0$, $\gamma^{(1)}$ and $\gamma^{(2)}$ are considered to be the same if $k_1=k_2$ and $\|c_i^{(1)}-c_i^{(2)}\|<\varepsilon$.

\paragraph*{Locality Sensitive Hashing Distance.}
Given a point set $Q\subset \R^2$, and $r>0$, It consider the disks with centers in $Q$ and radius equal to $r$. For LSH1$_Q$,
each trajectory is converted to a bit vector of length $|Q|$, and each bit represents the intersection of the trajectory with a disk.
and uses  Hamming distance of two bit vectors to define the distance between two curves.
For LSH2$_Q$,  each trajectories is converted to a sequence representing the
order in which the trajectory enters and exits the disks, and uses edit distance of two sequence to define the distance between two curves ~\cite{ACKGS2018}.

\section{More Trajectory Experiments}
\label{sec:more-exp}

\subsection{Different Weightings}
\label{SVM_linear_quadratic_synthetic_data}

For SVM with linear kernel and quadratic kernel, as we increase $w_1$, the error of $\dQ$ and $\dQ^\pi$
also decreases, although  not as obvious
as Gaussian kernel. The results are shown in Table \ref{table SVM synthetic data linear}
and Table \ref{table SVM synthetic data quadratic}.

\begin{table}[htbp]
	\begin{small}
		\center
		\caption{\small Landmark-sensitive classification error with weighted linear SVM.}
	
		\label{table SVM synthetic data linear}
		\begin{tabular}{c|ccc}
			\hline
			metrics &  mean  &  median  & variance \\
			\hline
			$\dQW$ ($w_1=0.3$)  & 0.3309 & 0.3333 & 0.0070 \\
			$\dQW$ ($w_1=0.6$)  & 0.3083 & 0.3333 & 0.0104 \\
			$\dQW$ ($w_1=0.9$)  & 0.3051 & 0.3333 & 0.0119 \\
			\hline
			$\dQW^\pi$ ($w_1=0.3$) & 0.4936 & 0.5000 & 0.0082 \\
			$\dQW^\pi$ ($w_1=0.6$) & 0.4191 & 0.4444 & 0.0049 \\
			$\dQW^\pi$ ($w_1=0.9$) & 0.4104 & 0.3889 & 0.0048 \\
			\hline
			$\dQWk$ ($w_1=0.3$) & 0.4372 & 0.4444 & 0.0081 \\
			$\dQWk$ ($w_1=0.6$) & 0.4340 & 0.4444 & 0.0080 \\
			$\dQWk$ ($w_1=0.9$) & 0.4329 & 0.4444 & 0.0080 \\
			\hline
		\end{tabular}
		
	\end{small}
\end{table}

\begin{table}[htbp]
	\begin{small}
		\center
		\caption{\small 	Landmark-sensitive classification error with weighted quadratic SVM.}
		\label{table SVM synthetic data quadratic}
		\begin{tabular}{c|ccc}
			\hline
			metrics &  mean  &  median  & variance \\
			\hline
			$\dQW$ ($w_1=0.3$)  & 0.3309 & 0.3333 & 0.0070 \\
			$\dQW$ ($w_1=0.6$)  & 0.3084 & 0.3333 & 0.0104 \\
			$\dQW$ ($w_1=0.9$)  & 0.3051 & 0.3333 & 0.0119 \\
			\hline
			$\dQW^\pi$ ($w_1=0.3$) & 0.5302 & 0.5000 & 0.0098 \\
			$\dQW^\pi$ ($w_1=0.6$) & 0.5270 & 0.5000 & 0.0105 \\
			$\dQW^\pi$ ($w_1=0.9$) & 0.3909 & 0.3889 & 0.0060 \\
			\hline
			$\dQWk$  ($w_1=0.3$) & 0.4367 & 0.4444 & 0.0081 \\
			$\dQWk$ ($w_1=0.6$) & 0.4333 & 0.4444 & 0.0079 \\
			$\dQWk$ ($w_1=0.9$) & 0.4322 & 0.4444 & 0.0079 \\
			\hline
		\end{tabular}
		
	\end{small}
\end{table}

\subsection{The error of LCSS, EAR and LSH with Other Parameters in Section \ref{sec:traj-analysis}}
\label{The error of LCSS, EAR and LSH with Other Parameters}

In the experiment of Section \ref{sec:traj-analysis}, the computation of LCSS, EAR LSH1$_Q$ and LSH2$_Q$
involves some parameters, and we only give the result of best parameter for these distances.
In this section, we describe the change of error statics for these distances with different parameters, and show
how we obtain the best parameter in each experiment. We use bold font to mark the smallest mean error and the corresponding median and variance.


\begin{table*}[t]
	\center
	\caption{Mean error of LCSS in Table \ref{table KNN real data1} with different parameters.}
	\label{Mean error of LCSS, table KNN real data1}
	\vspace{-.1in}
	
	\begin{tabular}{|r|cccccccccccc|}
		\hline
\diagbox{$\delta$}{mean}{$\varepsilon$}&
	0.0010&  0.0050&  0.0100&  0.0150&  0.0200&  0.0250&  0.0300&  0.0350&  0.0400&  0.0450&  0.0500&  0.0550\\
		\hline
1&    0.1115&    0.0822&    0.0856&    0.0969&    0.1105&    0.1297&    0.1532&    0.1749&    0.1902&    0.2003&    0.2087&    0.2182\\
2&    0.0940&    0.0785&    0.0840&    0.0954&    0.1085&    0.1278&    0.1511&    0.1731&    0.1879&    0.1987&    0.2064&    0.2164\\
3&    0.0901&    0.0769&    0.0833&    0.0946&    0.1078&    0.1271&    0.1503&    0.1718&    0.1865&    0.1977&    0.2057&    0.2160\\
4&    0.0860&    0.0755&    0.0823&    0.0936&    0.1077&    0.1267&    0.1496&    0.1707&    0.1861&    0.1966&    0.2050&    0.2151\\
5&    0.0846&    0.0745&    0.0819&    0.0935&    0.1079&    0.1269&    0.1495&    0.1704&    0.1857&    0.1961&    0.2046&    0.2150\\
6&    0.0826&    0.0739&    0.0821&    0.0939&    0.1079&    0.1265&    0.1494&    0.1706&    0.1855&    0.1958&    0.2045&    0.2149\\
7&    0.0816&    0.0734&    0.0823&    0.0937&    0.1078&    0.1261&    0.1490&    0.1702&    0.1853&    0.1957&    0.2043&    0.2147\\
8&    0.0802&    0.0729&    0.0817&    0.0935&    0.1075&    0.1261&    0.1489&    0.1702&    0.1852&    0.1957&    0.2041&    0.2145\\
9&    0.0795&    0.0721&    0.0815&    0.0933&    0.1075&    0.1262&    0.1490&    0.1699&    0.1849&    0.1955&    0.2039&    0.2142\\
10&   0.0783&    \pmb{0.0714}&    0.0811&    0.0930&    0.1072&    0.1258&    0.1485&    0.1695&    0.1845&    0.1951&    0.2037&    0.2140\\
	\hline
	\end{tabular}
\end{table*}

\begin{table*}[t]
	\center
	\caption{Median error of LCSS in Table \ref{table KNN real data1} with different parameters.}
	\label{Medina error of LCSS, table KNN real data1}
	\vspace{-.1in}
	
	\begin{tabular}{|r|cccccccccccc|}
		\hline
\diagbox{$\delta$}{mean}{$\varepsilon$}&
    0.0010&  0.0050&  0.0100&  0.0150&  0.0200&  0.0250&  0.0300&  0.0350&  0.0400&  0.0450&  0.0500&  0.0550\\
		\hline
1&    0.0869&    0.0577&    0.0589&    0.0652&    0.0741&    0.0889&    0.1029&    0.1143&    0.1240&    0.1325&    0.1387&    0.1474\\
2&    0.0707&    0.0536&    0.0576&    0.0643&    0.0722&    0.0868&    0.1000&    0.1118&    0.1222&    0.1304&    0.1360&    0.1458\\
3&    0.0667&    0.0531&    0.0571&    0.0640&    0.0720&    0.0867&    0.1000&    0.1103&    0.1200&    0.1297&    0.1357&    0.1464\\
4&    0.0625&    0.0526&    0.0565&    0.0640&    0.0720&    0.0864&    0.1000&    0.1094&    0.1200&    0.1278&    0.1353&    0.1449\\
5&    0.0608&    0.0524&    0.0564&    0.0643&    0.0728&    0.0865&    0.1000&    0.1087&    0.1194&    0.1274&    0.1357&    0.1449\\
6&    0.0590&    0.0516&    0.0567&    0.0649&    0.0729&    0.0857&    0.1000&    0.1088&    0.1189&    0.1267&    0.1353&    0.1449\\
7&    0.0583&    0.0512&    0.0571&    0.0647&    0.0728&    0.0857&    0.0987&    0.1088&    0.1187&    0.1267&    0.1353&    0.1446\\
8&    0.0571&    0.0506&    0.0568&    0.0646&    0.0730&    0.0857&    0.0984&    0.1083&    0.1187&    0.1266&    0.1346&    0.1444\\
9&    0.0566&    0.0500&    0.0564&    0.0645&    0.0731&    0.0857&    0.0984&    0.1079&    0.1182&    0.1261&    0.1340&    0.1444\\
10&   0.0556&    \pmb{0.0500}&    0.0563&    0.0643&    0.0728&    0.0850&    0.0976&    0.1076&    0.1179&    0.1256&    0.1336&    0.1440\\
		\hline
	\end{tabular}

\end{table*}

\begin{table*}[t]
	\center
	\caption{Error variance of LCSS in Table \ref{table KNN real data1} with different parameters.}
	\label{Error variance of LCSS, table KNN real data1}
	\vspace{-.1in}
	
	\begin{tabular}{|r|cccccccccccc|}
		\hline
\diagbox{$\delta$}{variance}{$\varepsilon$}&
    0.0010&  0.0050&  0.0100&  0.0150&  0.0200&  0.0250&  0.0300&  0.0350&  0.0400&  0.0450&  0.0500&  0.0550\\
		\hline
1&    0.0087&    0.0070&    0.0077&    0.0099&    0.0127&    0.0162&    0.0220&    0.0280&    0.0315&    0.0336&    0.0350&    0.0367\\
2&    0.0073&    0.0068&    0.0075&    0.0096&    0.0123&    0.0161&    0.0218&    0.0276&    0.0311&    0.0335&    0.0349&    0.0365\\
3&    0.0072&    0.0063&    0.0074&    0.0094&    0.0121&    0.0160&    0.0216&    0.0274&    0.0310&    0.0333&    0.0348&    0.0364\\
4&    0.0068&    0.0060&    0.0072&    0.0092&    0.0121&    0.0158&    0.0216&    0.0273&    0.0310&    0.0332&    0.0346&    0.0363\\
5&    0.0068&    0.0058&    0.0071&    0.0091&    0.0120&    0.0158&    0.0215&    0.0273&    0.0309&    0.0332&    0.0346&    0.0363\\
6&    0.0065&    0.0058&    0.0071&    0.0090&    0.0120&    0.0158&    0.0215&    0.0273&    0.0309&    0.0332&    0.0346&    0.0363\\
7&    0.0065&    0.0057&    0.0071&    0.0090&    0.0120&    0.0158&    0.0214&    0.0272&    0.0308&    0.0331&    0.0346&    0.0363\\
8&    0.0063&    0.0056&    0.0070&    0.0090&    0.0120&    0.0158&    0.0215&    0.0272&    0.0308&    0.0331&    0.0346&    0.0363\\
9&    0.0063&    0.0056&    0.0070&    0.0090&    0.0120&    0.0158&    0.0215&    0.0273&    0.0308&    0.0332&    0.0346&    0.0363\\
10&   0.0060&    \pmb{0.0054}&    0.0070&    0.0089&    0.0119&    0.0157&    0.0214&    0.0272&    0.0308&    0.0332&    0.0346&    0.0363\\
		\hline
	\end{tabular}

\end{table*}

\begin{table*}[t]
	\center
	\caption{Classification Error of EDR in Table \ref{table KNN real data1} with different parameters.}
	\label{Error of EDR, table KNN real data1}
	\vspace{-.1in}
	
	\begin{tabular}{r|cccccccccccc}
		\hline
$\varepsilon$ & 0.0010&  0.0050&  0.0100&  0.0150&  0.0200&  0.0250&  0.0300&  0.0350&  0.0400&  0.0450&  0.0500&  0.0550\\
		\hline
mean          & 0.1070&  \pmb{0.0802}&  0.0846&  0.0957&  0.1096&  0.1289&  0.1521&  0.1744&  0.1894&  0.1997&  0.2078&  0.2175\\
median        & 0.0833&  \pmb{0.0554}&  0.0581&  0.0640&  0.0731&  0.0875&  0.1009&  0.1139&  0.1229&  0.1319&  0.1378&  0.1462\\
variance      & 0.0084&  \pmb{0.0070}&  0.0077&  0.0098&  0.0127&  0.0162&  0.0219&  0.0279&  0.0314&  0.0336&  0.0350&  0.0367\\
		\hline
		
	\end{tabular}

\end{table*}

\begin{table*}[t]
	\center
	\caption{Classification Error of LSH1$_Q$ and LSH2$_Q$ in Table \ref{table KNN real data1} with different parameters.}
	\label{Error of LSH, table KNN real data1}
	\vspace{-.1in}
	
	\begin{tabular}{rr|cccccccccccc}
		\hline
		 &$r$ & 0.0050& 0.0100 & 0.0200 & 0.0300 & 0.0400 & 0.0500 & 0.0600  &  0.0700 & 0.0800 & 0.0900 & 0.1000 &0.1100\\
		 \hline
         &mean    &  0.4145& 0.3792&  0.3143&   0.2645&  0.2197&  0.1374&  \pmb{0.1290}&  0.1501&  0.1487&  0.1774&  0.1680&  0.1633\\
LSH1$_Q$ &median  &  0.3913& 0.3500&  0.2693&   0.2121&  0.1667&  0.1000&  \pmb{0.0949}&  0.1114&  0.1046&  0.1133&  0.1154&  0.1179\\
		 &variance&  0.0616& 0.0530&  0.0448&   0.0404&  0.0315&  0.0153&  \pmb{0.0128}&  0.0168&  0.0176&  0.0286&  0.0232&  0.0207\\
		 \hline
		 &mean    &  0.4161&  0.3873&  0.3449&  0.3043&  0.2798&  0.2637&  0.2574&  0.2494&  0.2445&  0.2415&  \pmb{0.2409}&  0.2426\\
LSH2$_Q$ &median  &  0.3919&  0.3644&  0.3154&  0.2605&  0.2333&  0.2281&  0.2255&  0.2275&  0.2216&  0.2195&  \pmb{0.2182}&  0.2191\\
		 &variance&  0.0621&  0.0531&  0.0457&  0.0405&  0.0381&  0.0301&  0.0271&  0.0224&  0.0220&  0.0210&  \pmb{0.0210}&  0.0215\\
		 \hline
	\end{tabular}

\end{table*}

\begin{table*}[t]
	\center
	\caption{Mean error of LCSS in Table \ref{table KNN and SVM real data2 Q1} with different parameters.}
	\label{Mean error of LCSS, table KNN real data2 Q1}
	\vspace{-.1in}
	\begin{tabular}{|r|cccccccccccc|}
	\hline
	\diagbox{$\delta$}{mean}{$\varepsilon$}
 &  0.0010&  0.0050&  0.0100&  0.0150&  0.0200&  0.0250&  0.0300&  0.0350&  0.0400&  0.0450&  0.0500&  0.0550\\
 	\hline
1&  0.2761&  0.2956&  0.2634&  0.2676&  0.2839&  0.3046&  0.3049&  0.3147&  0.3256&  0.3459&  0.3527&  0.3582\\
2&  0.3123&  0.2761&  0.2718&  0.2647&  0.2866&  0.3148&  0.3112&  0.3160&  0.3267&  0.3458&  0.3552&  0.3598\\
3&  0.3116&  0.2905&  0.2685&  \pmb{0.2448}&  0.2941&  0.3350&  0.3129&  0.3160&  0.3267&  0.3458&  0.3542&  0.3598\\
4&  0.3112&  0.2911&  0.2552&  0.2556&  0.2937&  0.3347&  0.3135&  0.3160&  0.3267&  0.3458&  0.3542&  0.3598\\
5&  0.2823&  0.2919&  0.2680&  0.2656&  0.2965&  0.3352&  0.3135&  0.3160&  0.3267&  0.3458&  0.3542&  0.3598\\
6&  0.2883&  0.2904&  0.2691&  0.2726&  0.2965&  0.3352&  0.3135&  0.3160&  0.3267&  0.3458&  0.3542&  0.3596\\
7&  0.2931&  0.2887&  0.2645&  0.2726&  0.2965&  0.3352&  0.3135&  0.3160&  0.3267&  0.3458&  0.3542&  0.3598\\
8&  0.2952&  0.2828&  0.2655&  0.2726&  0.2965&  0.3352&  0.3135&  0.3160&  0.3267&  0.3458&  0.3542&  0.3598\\
19& 0.2946&  0.2831&  0.2655&  0.2726&  0.2965&  0.3352&  0.3135&  0.3160&  0.3267&  0.3458&  0.3542&  0.3598\\
10& 0.2934&  0.2831&  0.2655&  0.2726&  0.2965&  0.3352&  0.3135&  0.3160&  0.3267&  0.3458&  0.3542&  0.3598\\	
	\hline
\end{tabular}

\end{table*}

  \begin{table*}[t]
  	\center
  	\caption{Median error of LCSS in Table \ref	{table KNN and SVM real data2 Q1} with different parameters.}
  	\label{Median error of LCSS, table KNN real data2 Q1}
  	\vspace{-.1in}
  	\begin{tabular}{|r|cccccccccccc|}
  		\hline
	\diagbox{$\delta$}{median}{$\varepsilon$}
 &  0.0010&  0.0050&  0.0100&  0.0150&  0.0200&  0.0250&  0.0300&  0.0350&  0.0400&  0.0450&  0.0500&  0.0550\\
		\hline
1&  0.2778&  0.3056&  0.2500&  0.2778&  0.2778&  0.3056&  0.3056&  0.3056&  0.3333&  0.3611&  0.3611&  0.3611\\
2&  0.3056&  0.2778&  0.2778&  0.2500&  0.2778&  0.3056&  0.3056&  0.3056&  0.3333&  0.3611&  0.3611&  0.3611\\
3&  0.3056&  0.2778&  0.2778&  \pmb{0.2500}&  0.2778&  0.3333&  0.3056&  0.3056&  0.3333&  0.3611&  0.3611&  0.3611\\
4&  0.3056&  0.2778&  0.2500&  0.2500&  0.2778&  0.3333&  0.3056&  0.3056&  0.3333&  0.3611&  0.3611&  0.3611\\
5&  0.2778&  0.2778&  0.2778&  0.2500&  0.3056&  0.3333&  0.3056&  0.3056&  0.3333&  0.3611&  0.3611&  0.3611\\
6&  0.2778&  0.2778&  0.2778&  0.2778&  0.3056&  0.3333&  0.3056&  0.3056&  0.3333&  0.3611&  0.3611&  0.3611\\
7&  0.3056&  0.2778&  0.2500&  0.2778&  0.3056&  0.3333&  0.3056&  0.3056&  0.3333&  0.3611&  0.3611&  0.3611\\
8&  0.3056&  0.2778&  0.2500&  0.2778&  0.3056&  0.3333&  0.3056&  0.3056&  0.3333&  0.3611&  0.3611&  0.3611\\
9&  0.3056&  0.2778&  0.2500&  0.2778&  0.3056&  0.3333&  0.3056&  0.3056&  0.3333&  0.3611&  0.3611&  0.3611\\
10& 0.3056&  0.2778&  0.2500&  0.2778&  0.3056&  0.3333&  0.3056&  0.3056&  0.3333&  0.3611&  0.3611&  0.3611\\
  		\hline
  	\end{tabular}
  \end{table*}

\begin{table*}[t]
	\center
	\caption{Error variance of LCSS in Table \ref{table KNN and SVM real data2 Q1} with different parameters.}
	\label{Error variance of LCSS, table KNN real data2 Q1}
	\vspace{-.1in}
	\begin{tabular}{|r|cccccccccccc|}
		\hline
	\diagbox{$\delta$}{variance}{$\varepsilon$}
 &  0.0010&  0.0050&  0.0100&  0.0150&  0.0200&  0.0250&  0.0300&  0.0350&  0.0400&  0.0450&  0.0500&  0.0550\\
		\hline
1&  0.0030&  0.0037&  0.0038&  0.0041&  0.0036&  0.0034&  0.0025&  0.0013&  0.0008&  0.0004&  0.0002&  0.0003\\
2&  0.0034&  0.0035&  0.0038&  0.0038&  0.0036&  0.0038&  0.0026&  0.0012&  0.0007&  0.0003&  0.0002&  0.0003\\
3&  0.0030&  0.0035&  0.0038&  \pmb{0.0037}&  0.0037&  0.0038&  0.0025&  0.0012&  0.0007&  0.0003&  0.0002&  0.0003\\
4&  0.0030&  0.0036&  0.0035&  0.0037&  0.0036&  0.0038&  0.0025&  0.0012&  0.0007&  0.0003&  0.0002&  0.0003\\
5&  0.0031&  0.0035&  0.0037&  0.0039&  0.0035&  0.0038&  0.0025&  0.0012&  0.0007&  0.0003&  0.0002&  0.0003\\
6&  0.0029&  0.0033&  0.0036&  0.0039&  0.0035&  0.0038&  0.0025&  0.0012&  0.0007&  0.0003&  0.0002&  0.0003\\
7&  0.0029&  0.0034&  0.0035&  0.0040&  0.0035&  0.0038&  0.0025&  0.0012&  0.0007&  0.0003&  0.0002&  0.0003\\
8&  0.0030&  0.0031&  0.0035&  0.0040&  0.0035&  0.0038&  0.0025&  0.0012&  0.0007&  0.0003&  0.0002&  0.0003\\
9&  0.0029&  0.0031&  0.0035&  0.0040&  0.0035&  0.0038&  0.0025&  0.0012&  0.0007&  0.0003&  0.0002&  0.0003\\
10& 0.0029&  0.0031&  0.0035&  0.0040&  0.0035&  0.0038&  0.0025&  0.0012&  0.0007&  0.0003&  0.0002&  0.0003\\
		\hline
	\end{tabular}
\end{table*}

\begin{table*}[t]
	\center
	\caption{Classification error of EDR in Table \ref{table KNN and SVM real data2 Q1} with different parameters.}
	\label{Classification error of EDR, table KNN real data2 Q1}
	\vspace{-.1in}
	
	\begin{tabular}{r|cccccccccccccc}
	\hline
$\varepsilon$ & 0.0010&  0.0050&  0.0100&  0.0150&  0.0200&  0.0250&  0.0300&  0.0350&  0.0400&  0.0450&  0.0500&  0.0550\\
	\hline
mean          & 0.2748&  0.2932&  0.2661&  \pmb{0.2640}&  0.2854&  0.3036&  0.3050&  0.3147&  0.3256&  0.3459&  0.3527&  0.3582\\
median        & 0.2778&  0.2778&  0.2639&  \pmb{0.2500}&  0.2778&  0.3056&  0.3056&  0.3056&  0.3333&  0.3611&  0.3611&  0.3611\\
variance      & 0.0028&  0.0038&  0.0037&  \pmb{0.0039}&  0.0035&  0.0035&  0.0025&  0.0013&  0.0008&  0.0004&  0.0002&  0.0003\\
	\hline	
\end{tabular}
\end{table*}

\begin{table*}[t]
	\center
	\caption{Classification error of  LSH1$_Q$ and LSH2$_Q$ in Table \ref{table KNN and SVM real data2 Q1} with different parameters.}
	\label{Classification error of LSH, table KNN real data2 Q1}
	\vspace{-.1in}
	
	\begin{tabular}{rr|cccccccccccc}
	\hline
            &$r$      &0.0050&  0.0100&        0.0200&  0.0300&  0.0400&  0.0500&  0.0600&  0.0700&  0.0800&  0.0900&  0.1000&  0.1100\\
    \hline
            &mean     &0.3360&  0.2767&  \pmb{0.2673}&  0.2784&  0.3211&  0.3804&  0.3647&  0.3707&  0.3627&  0.3616&  0.3611&  0.3659\\
LSH1$_{Q_1}$&median   &0.3611&  0.2778&  \pmb{0.2778}&  0.2778&  0.3333&  0.3611&  0.3611&  0.3611&  0.3611&  0.3611&  0.3611&  0.3611\\
            &variance &0.0013&  0.0026&  \pmb{0.0020}&  0.0028&  0.0038&  0.0012&  0.0007&  0.0012&  0.0002&  0.0000&  0.0000&  0.0004\\
	\hline
	        &mean     &0.3361&  0.2959&  0.2789&  0.2997&  0.2869&  0.2830&  0.2811&  0.2543&  \pmb{0.2516}&  0.2619&  0.2684&  0.2834\\
LSH2$_{Q_1}$&median   &0.3611&  0.3056&  0.2778&  0.3056&  0.2778&  0.2778&  0.2778&  0.2500&  \pmb{0.2500}&  0.2778&  0.2778&  0.2778\\
			&variance &0.0013&  0.0020&  0.0016&  0.0026&  0.0027&  0.0031&  0.0022&  0.0021&  \pmb{0.0022}&  0.0017&  0.0016&  0.0014\\
	\hline
			&mean     &0.3365&  0.2517&  0.2352&  \pmb{0.2209}&  0.2468&  0.2642&  0.3334&  0.3020&  0.3754&  0.3668&  0.3626&  0.3611\\
LSH1$_{Q_2}$&median   &0.3333&  0.2500&  0.2222&  \pmb{0.2222}&  0.2500&  0.2500&  0.3333&  0.3056&  0.3611&  0.3611&  0.3611&  0.3611\\
			&variance &0.0007&  0.0026&  0.0037&  \pmb{0.0039}&  0.0035&  0.0027&  0.0028&  0.0024&  0.0009&  0.0005&  0.0004&  0.0000\\
	\hline
         	&mean     &0.3472&  0.3480&  0.3428&  0.2879&  0.3131&  \pmb{0.2690}&  0.2945&  0.2857&  0.3217&  0.3072&  0.3249&  0.3164\\
LSH2$_{Q_2}$&median   &0.3611&  0.3611&  0.3333&  0.2778&  0.3056&  \pmb{0.2778}&  0.3056&  0.2778&  0.3333&  0.3056&  0.3333&  0.3056\\
			&variance &0.0008&  0.0006&  0.0022&  0.0018&  0.0028&  \pmb{0.0022}&  0.0023&  0.0023&  0.0015&  0.0021&  0.0022&  0.0015\\
\hline	
\end{tabular}
\end{table*}

\begin{table*}[t]
	\center
	\caption{Mean error of LCSS in Table \ref{table KNN the synthetic_data} with different parameters.}
	\label{Mean error of LCSS, table KNN the synthetic_data}
	\vspace{-.1in}
	\begin{tabular}{|r|cccccccccccc|}
		\hline
		\diagbox{$\delta$}{mean}{$\varepsilon$}
 &  0.0010&  0.0050&  0.0100&  0.0150&  0.0200&  0.0250&  0.0300&  0.0350&  0.0400&  0.0450&  0.0500&  0.0550\\
		\hline
1&  0.4961&  0.4776&  0.4484&  0.4584&  0.4068&  0.4395&  0.4412&  0.4033&  0.4233&  0.4585&  0.5073&  0.5142\\
2&  0.4112&  0.3867&  0.4405&  0.4520&  0.4363&  0.4539&  0.4238&  0.4611&  0.4999&  0.5007&  0.5062&  0.5299\\
3&  0.4025&  0.4163&  0.4903&  0.4728&  0.4343&  0.4307&  0.4389&  0.4448&  0.4656&  0.4737&  0.4814&  0.5158\\
4&  0.3504&  0.4115&  0.4320&  0.4481&  0.4435&  0.4066&  0.4319&  0.4546&  0.4397&  0.4511&  0.4631&  0.4901\\
5&  0.3509&  0.4190&  0.4082&  0.4217&  0.4177&  0.4061&  0.4378&  0.4453&  0.4606&  0.4389&  0.4738&  0.4983\\
6&  0.3481&  0.4117&  0.3961&  0.4000&  0.3939&  0.4045&  0.4368&  0.4465&  0.4592&  0.4391&  0.4759&  0.4993\\
7&  0.3527&  0.4241&  0.3996&  0.4009&  0.3947&  0.4071&  0.4308&  0.4387&  0.4569&  0.4397&  0.4780&  0.4993\\
8&  \pmb{0.3437}&  0.4141&  0.3998&  0.4009&  0.3947&  0.4064&  0.4308&  0.4324&  0.4554&  0.4390&  0.4780&  0.4993\\
9&  0.3499&  0.4244&  0.4039&  0.3969&  0.3961&  0.4064&  0.4308&  0.4324&  0.4554&  0.4390&  0.4780&  0.4993\\
10& 0.3582&  0.4329&  0.4041&  0.3969&  0.3961&  0.4064&  0.4308&  0.4324&  0.4554&  0.4390&  0.4780&  0.4993\\
		\hline
	\end{tabular}
\end{table*}

\begin{table*}[t]
	\center
	\caption{Median error of LCSS in Table \ref{table KNN the synthetic_data} with different parameters.}
	\label{Median error of LCSS, table KNN the synthetic_data}
	\vspace{-.1in}
	\begin{tabular}{|r|cccccccccccc|}
		\hline
		\diagbox{$\delta$}{median}{$\varepsilon$}
 &  0.0010&  0.0050&  0.0100&  0.0150&  0.0200&  0.0250&  0.0300&  0.0350&  0.0400&  0.0450&  0.0500&  0.0550\\
 		\hline
1&  0.5000&  0.5000&  0.4444&  0.4444&  0.3889&  0.4444&  0.4444&  0.3889&  0.4444&  0.4444&  0.5000&  0.5000\\
2&  0.3889&  0.3889&  0.4444&  0.4444&  0.4444&  0.4444&  0.4444&  0.4444&  0.5000&  0.5000&  0.5000&  0.5000\\
3&  0.3889&  0.4444&  0.5000&  0.4444&  0.4444&  0.4444&  0.4444&  0.4444&  0.4444&  0.4444&  0.5000&  0.5000\\
4&  0.3333&  0.3889&  0.4444&  0.4444&  0.4444&  0.3889&  0.4444&  0.4444&  0.4444&  0.4444&  0.4444&  0.5000\\
5&  0.3333&  0.4444&  0.3889&  0.4444&  0.4444&  0.3889&  0.4444&  0.4444&  0.4444&  0.4444&  0.5000&  0.5000\\
6&  0.3333&  0.3889&  0.3889&  0.3889&  0.3889&  0.3889&  0.4444&  0.4444&  0.4444&  0.4444&  0.5000&  0.5000\\
7&  0.3333&  0.4444&  0.3889&  0.3889&  0.3889&  0.3889&  0.4444&  0.4444&  0.4444&  0.4444&  0.5000&  0.5000\\
8&  \pmb{0.3333}&  0.4167&  0.3889&  0.3889&  0.3889&  0.3889&  0.4444&  0.4444&  0.4444&  0.4444&  0.5000&  0.5000\\
9&  0.3333&  0.4444&  0.3889&  0.3889&  0.3889&  0.3889&  0.4444&  0.4444&  0.4444&  0.4444&  0.5000&  0.5000\\
10& 0.3333&  0.4444&  0.3889&  0.3889&  0.3889&  0.3889&  0.4444&  0.4444&  0.4444&  0.4444&  0.5000&  0.5000\\
		\hline
	\end{tabular}

\end{table*}

\begin{table*}[t]
	\center
	\caption{Error variance of LCSS in Table \ref{table KNN the synthetic_data} with different parameters.}
	\label{Error variance of LCSS, table KNN the synthetic_data}
	\vspace{-.1in}
	\begin{tabular}{|r|cccccccccccc|}
		\hline
		\diagbox{$\delta$}{variance}{$\varepsilon$}
 &  0.0010&  0.0050&  0.0100&  0.0150&  0.0200&  0.0250&  0.0300&  0.0350&  0.0400&  0.0450&  0.0500&  0.0550\\
 		\hline
1&  0.0015&  0.0036&  0.0064&  0.0062&  0.0070&  0.0068&  0.0066&  0.0058&  0.0061&  0.0068&  0.0076&  0.0077\\
2&  0.0051&  0.0063&  0.0066&  0.0070&  0.0074&  0.0062&  0.0069&  0.0069&  0.0069&  0.0073&  0.0074&  0.0081\\
3&  0.0053&  0.0058&  0.0071&  0.0070&  0.0068&  0.0056&  0.0069&  0.0069&  0.0079&  0.0084&  0.0079&  0.0093\\
4&  0.0058&  0.0065&  0.0073&  0.0072&  0.0073&  0.0058&  0.0074&  0.0072&  0.0078&  0.0080&  0.0077&  0.0090\\
5&  0.0069&  0.0078&  0.0077&  0.0073&  0.0063&  0.0059&  0.0072&  0.0072&  0.0076&  0.0078&  0.0078&  0.0091\\
6&  0.0070&  0.0075&  0.0077&  0.0075&  0.0067&  0.0059&  0.0074&  0.0070&  0.0072&  0.0078&  0.0077&  0.0090\\
7&  0.0066&  0.0072&  0.0076&  0.0079&  0.0070&  0.0058&  0.0072&  0.0068&  0.0070&  0.0078&  0.0076&  0.0090\\
8&  \pmb{0.0066}&  0.0074&  0.0076&  0.0079&  0.0070&  0.0057&  0.0072&  0.0063&  0.0069&  0.0077&  0.0076&  0.0090\\
9&  0.0066&  0.0072&  0.0076&  0.0078&  0.0068&  0.0057&  0.0072&  0.0063&  0.0069&  0.0077&  0.0076&  0.0090\\
10& 0.0066&  0.0070&  0.0076&  0.0078&  0.0068&  0.0057&  0.0072&  0.0063&  0.0069&  0.0077&  0.0076&  0.0090\\
		\hline
	\end{tabular}

\end{table*}

\begin{table*}[t]
	\center
	\caption{Classification error of EDR in Table \ref{table KNN the synthetic_data} with different parameters.}
	\label{Classification error of EDR, table KNN the synthetic_data}
	\vspace{-.1in}
	\begin{tabular}{r|cccccccccccc}
		\hline
$\varepsilon$ & 0.0010&  0.0050&  0.0100&  0.0150&  0.0200&  0.0250&  0.0300&  0.0350&  0.0400&  0.0450&  0.0500&  0.0550\\
		\hline
mean          & 0.4632&  0.4541&  0.4171&  0.4450&  \pmb{0.3916}&  0.4134&  0.4422&  0.4259&  0.4618&  0.4559&  0.4681&  0.5123\\
median		  & 0.4444&  0.4444&  0.4444&  0.4444&  \pmb{0.3889}&  0.3889&  0.4444&  0.4444&  0.4444&  0.4444&  0.4444&  0.5000\\
variance      & 0.0026&  0.0043&  0.0065&  0.0068&  \pmb{0.0068}&  0.0064&  0.0058&  0.0060&  0.0062&  0.0079&  0.0081&  0.0087\\
		\hline
	\end{tabular}

\end{table*}

\begin{table*}[t]
	\center
	\caption{Classification Error of LSH1$_Q$ and LSH2$_Q$ in Table \ref{table KNN the synthetic_data} with different parameters.}
	\label{Error of LSH, table KNN the synthetic_data}
	\vspace{-.1in}
	
	\begin{tabular}{rr|cccccccccccc}
		\hline
		&$r$ & 0.0050 & 0.0100 & 0.0200 & 0.0300 & 0.0400 & 0.0500 & 0.0600  &  0.0700 & 0.0800 & 0.0900 & 0.1000 & 0.1100\\
		\hline
		 &mean     & 0.5098&  \pmb{0.2524}&  0.2950&  0.4878&  0.4443&  0.4691&  0.4494&  0.4558&  0.5046&  0.5103&  0.4439&  0.4305\\
LSH1$_Q$ &median   & 0.5000&  \pmb{0.2222}&  0.2778&  0.5000&  0.4444&  0.4444&  0.4444&  0.4444&  0.5000&  0.5000&  0.4444&  0.4444\\
		 &variance & 0.0006&  \pmb{0.0098}&  0.0067&  0.0064&  0.0062&  0.0066&  0.0085&  0.0059&  0.0045&  0.0068&  0.0068&  0.0061\\
		\hline
	  	 &mean     &0.5000&  0.4547&  \pmb{0.3248}&  0.3850&  0.5271&  0.5400&  0.5216&  0.5130&  0.4828&  0.4943&  0.4406&  0.4865\\
LSH2$_Q$ &median   &0.5000&  0.4444&  \pmb{0.3333}&  0.3889&  0.5278&  0.5556&  0.5000&  0.5000&  0.5000&  0.5000&  0.4444&  0.5000\\
		 &variance &0     &  0.0074&  \pmb{0.0084}&  0.0076&  0.0068&  0.0049&  0.0046&  0.0072&  0.0052&  0.0049&  0.0076&  0.0070\\

		\hline
	\end{tabular}
\end{table*}

\end{document}